\documentclass[onecolumn,draft, 11pt]{IEEEtran}

\usepackage{setspace}
\usepackage{color}

\usepackage{graphicx,tabularx,array,geometry,amsmath,amsthm,thmtools,caption}

\usepackage{mathtools}
\DeclarePairedDelimiter{\ceil}{\lceil}{\rceil}

\usepackage{caption}
\usepackage{amsfonts}
\usepackage{bbm}
\usepackage{makecell}
\usepackage{multirow}
 \usepackage{amssymb}
\usepackage{amsthm}
\usepackage{txfonts}
\usepackage[T1]{fontenc}
\usepackage{tikz}
\usepackage[scr=dutchcal]{mathalfa}
\let\mathscr\mathbscr

\DeclareMathOperator*{\argmin}{arg\,min}

\newcolumntype{x}[1]{>{\centering\arraybackslash}p{#1}}

\newcommand{\indep}{\rotatebox[origin=c]{90}{$\models$}}

\declaretheoremstyle[headfont=\bfseries, 
    bodyfont=\normalfont]{normalhead}
\declaretheorem[style=normalhead]{Example}

\newtheorem{Theorem}{Theorem}
\newtheorem{Proposition}{Proposition}
\newtheorem{Lemma}{Lemma}

\newtheorem{Remark}{Remark}
\newtheorem{Definition}{Definition}

\newtheorem{Claim}{Claim}

\allowdisplaybreaks 

\begin{document}

\title{A New Achievable Rate-Distortion Region for Distributed Source Coding}
\author{Farhad Shirani  and S. Sandeep Pradhan\\
Dept. of Electrical Engineering and Computer Science \\
 Univ. of  Michigan, Ann Arbor, MI. \\\date{} }

\maketitle \thispagestyle{empty} \pagestyle{plain}

\begin{abstract}
In this work, lossy distributed compression of a pair of correlated sources is considered. Conventionally, Shannon's random coding arguments --- using randomly generated unstructured codebooks whose  blocklength is taken to be asymptotically large --- are used to derive achievability results. However,  in some multi-terminal communications scenarios, using random codes with constant finite blocklength in certain coding architectures may lead to improved achievable regions compared to the conventional approach. In other words, in some network communication scenarios, there is a finite optimal value in the blocklength of the randomly generated code used for distributed processing of information sources. Motivated by this, a coding scheme is proposed which consists of two codebook layers: i) the primary codebook which has constant finite blocklength, and ii) the secondary codebook whose blocklength is taken to be asymptotically large. 
The achievable performance  is analyzed in two steps. In the first step, a characterization of an inner bound to the achievable region is derived in terms information measures which are functions of multi-letter probability distributions. In the next step, a computable single-letter inner-bound to the achievable region is extracted. It is shown through an example that the resulting rate-distortion region is strictly larger than the Berger-Tung achievable region. 
\end{abstract}
\section{introduction}
Distributed source coding (DSC) involves the compression and storage of several physically separated, statistically correlated source sequences into a shared storage unit. \textcolor{black}{In the classical lossy distributed source coding problem, shown in Figure \ref{Fig: DSC}, two distributed encoders observe a pair of correlated sources and communicate a compressed version of their respective source sequences to a joint decoder.} The decoder then wishes to produce a faithful reconstruction of the two sources relative to a fidelity criterion. Distributed source coding techniques are widely used in various practical settings such as video coding \cite{puri2007prism, girod2005distributed}, sensor networks \cite{xiong2004distributed,akyildiz2007wireless}, and relay-based data transmission systems \cite{aljohani2016distributed}.

 The correlation between the distributed information sources is a key resource in the DSC problem. \textcolor{black}{Absent any correlation (i.e. $X_1~ \indep ~ X_2$), the distributed encoders are unable to collaborate with each other, and the optimal rate-distortion performance is that of two separate point-to-point lossy source compression schemes \cite{ElGamalLec}.} In the other extreme, when the two sources are fully correlated (i.e. $X_1=X_2$), the encoders can operate with complete collaboration, and achieve the same rate-distortion performance as a centralized encoder.
 A well-designed DSC scheme leverages the correlation between the pair of source sequences to facilitate collaboration between the distributed encoders,  and achieve a reliable reconstruction of the sources at the joint decoder while minimizing the transmission rates of the encoders. Our objective is to design such a DSC scheme and to characterize the optimal rate-distortion trade-off. 

\begin{figure}
\includegraphics*[draft=false,scale=.8]{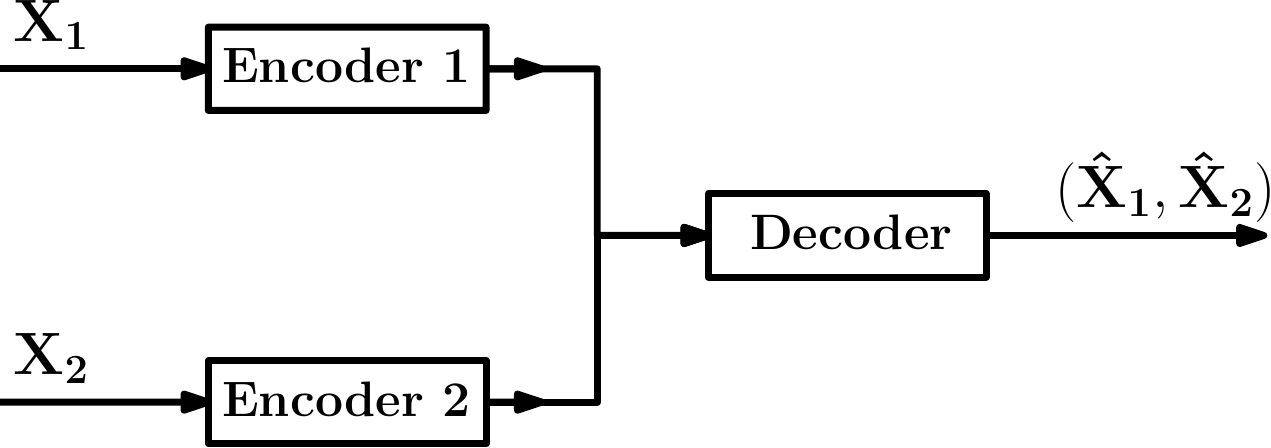}
\centering 
\caption{General Lossy Distributed Source Coding}
\label{Fig: DSC}
\end{figure}

Slepian and Wolf \cite{SW} laid the foundations of DSC by characterizing the optimal rate region in the lossless reconstruction regime which is achieved using a random binning scheme. \textcolor{black}{Building upon this, Berger and Tung \cite{Markov} considered the general lossy DSC setup and introduced a coding scheme, which is now referred to as the `\textit{Berger-Tung}' (BT) coding scheme. A characterization of the asymptotic  performance of the BT scheme results in the best-known inner bound to the optimal rate-distortion (RD) region. In this strategy, the two encoders use a pair of randomly and independently generated codebooks to quantize the source sequences using the conventional typicality-based  encoding method.} The outputs of these quantizers are binned randomly and independently, and the bin indices are transmitted to the decoder. The decoder reconstructs the quantized sequences by finding the unique pair of jointly typical sequences in the corresponding bins, where typicality is measured based on the joint distribution imposed on the quantized sequences through the distributed quantization process.\textcolor{black}{ The binning step reduces the transmission rates of the encoders  by leveraging the \textit{correlation} among the quantized sequences.}
Loosely speaking, higher correlation among the quantized sequences allows for more efficient binning, and leads to lower transmission rates. Consequently, the encoders in the BT scheme must preserve the correlation among the input sequences  during the quantization step. \textcolor{black}{However, independent codebook generation in the BT scheme leads to the so-called \textit{long Markov chain} which limits the correlation among the quantized sequences.} To elaborate, let $U_1$ and $U_2$ represent the single-letter random variables corresponding to the quantizations of $X_1$ and $X_2$, respectively. The long Markov chain $U_1 - X_1 - X_2 - U_2$ implies that conditioned on the sources, the single-letter distribution of the quantized versions of the sources decomposes into a product of conditional marginal distributions. \textcolor{black}{On the other hand, an outer bound to the achievable RD region for distributed source coding can be provided by modifying the single-letter characterization of the BT region \cite{berger1978multiterminal}, where 
primarily  the long Markov chain is replaced by a pair of short Markov chains $X_1 - X_2 - U_2$ and $U_1 - X_1 - X_2$.
}
This observation suggests that one approach to possibly improve upon the BT achievable region is to `\textit{break}' the long Markov chain.

\textcolor{black}{A scenario where the long Markov chain may be relaxed is when a common component is present in the distributed sources \cite{wagner}. A common component of the sources $X_1$ and $X_2$ is a random variable $V=f(X_1)=g(X_2)$ which can be computed independently at each of the distributed terminals. In \cite{wagner}, the `\textit{Common Component}' (CC) coding scheme was introduced, where at each terminal the common component sequence is quantized using  randomly generated identical codebooks. Consequently, both encoders compute the same quantized sequence. The encoders fully collaborate to send this sequence to the central decoder. The quantized sequence is treated as side-information available at all transmitter and receiver terminals, and the rest of the coding scheme is implemented in a similar fashion as the BT scheme.}  \textcolor{black} {This leads to an achievable region called the CC region, where the long Markov chain in the BT region is relaxed. To elaborate, let $W$ be the random variable corresponding the quantization of the common component sequence. The long Markov chain $U_1 - X_1 - X_2 - U_2$ in the BT achievable region is relaxed to $U_1 - (W,X_1) - (W,X_2) - U_2$ in the CC region. In the absence of common components, the CC scheme reduces to the BT scheme. It was shown that the CC achievable region is discontinuous with respect to the joint distribution $P_{X_1,X_2}$ of the sources \cite{wagner}.} The reason is that the common component is fragile, and slight perturbations in the source probability distribution $P_{X_1,X_2}$ can create or eliminate common components with large entropy. As a result, the CC achievable RD region shrinks discontinuously in source probability distribution as common components are replaced with highly correlated components. On the other hand, it is known that the optimal achievable region is continuous with respect to the joint source distribution. Therefore, it was pointed out in \cite{wagner} that the CC rate-distortion region cannot be optimal. However, it was not clear how to achieve points outside of the CC rate-distortion region.

As noted earlier, CC scheme uses identical codebooks for quantizing the common component in the distributed terminals. As a result, the two encoders compute the same quantization sequence, and there is no loss of correlation. \textcolor{black}{However, if common components are replaced by highly correlated components, then it is known that blockwise processing of the correlated sequences leads to a significant loss of correlation when the blocklength is asymptotically large.} It was first observed by G\'acs, K\"orner and Witsenhausen \cite{ComInf1,ComInf2} that coding over blocks 
decreases distributed correlation. To elaborate, consider the simplified version of the DSC problem where the objective of the distributed encoders is to encode the observations into a single bit. Let $e$ and $f$ denote the encoding functions associated with the two encoders. We wish to maximize correlation between the outputs such that $H(e(X^n))>0$ and  $H(f(X^n))>0$, where $H(\cdot)$ is the entropy function. It was shown that maximum correlation is achieved when the output depends only on one of the input samples at both encoders. In fact, any block mapping strictly reduces the correlation between the output bits. In summary, uncoded mappings (mappings with blocklength equals $1$) are optimal in terms of correlation preservation. This suggests that the use of random codes of constant finite blocklength may improve upon the BT coding scheme. In our preliminary work \cite{FinLen}, a new inner bound to the achievable RD region was derived which strictly improved upon the BT region. Building on this work, new coding strategies and achievable rate regions have been derived for transmission of sources over the multiple access channel \cite{padakandla2016communicating} and the interference channel \cite{padakandla2017communicating}. In \cite{chaharsooghi2017sub,chaharsooghi2017correlation,shirani2019sub}, we used these observations to prove the suboptimality of the conventional single-letter coding schemes in various multiterminal communication scenarios. This generalizes a result shown by Dueck \cite{dueck1981note} in the case of transmission of sources over MAC, where the suboptimality of the Cover-El Gamal-Salehi scheme \cite{cover1980multiple} was proved by providing a counter-example.

\textcolor{black}{
In this paper,  we provide a new coding scheme for the two user DSC problem. The scheme uses a novel two-layered coding approach, where the first code layer consists of randomly generated codes with constant finite blocklengths, and the second code layer consists of  randomly generated codes with asymptotically large blocklengths.} The two coding layers are interleaved using a new interleaving method which we call \textit{Finite Length Matrix Coding} (FLMC). The interleaving step is necessary in order to facilitate the performance analysis of the proposed scheme, and to derive a computable characterization of the resulting achievable RD region. Roughly speaking, the first coding layer, which utilizes codes with constant finite blocklength, is used to quantize highly correlated components of the sources. 
As the blocklength of this code layer is increased, the covering efficiency increases
(due to law of large numbers), while the correlation preserving efficiency decreases  (following 
G\'acs, K\"orner and Witsenhausen \cite{ComInf1,ComInf2}), and hence a trade-off 
between the two is manifested. \textcolor{black}{
There is a sweet spot for the blocklength where the overall system efficiency achieves its maximum. 
}

In the extreme case when the highly correlated components are common components, the blocklength in the first layer is taken to be asymptotically large, and the scheme transposes into the CC scheme. The second coding layer is similar to the 
codes used in the BT scheme. The rate-distortion region of our proposed scheme is first characterized using information measures expressed as functions of multi-letter distributions. In the next step, a computable single-letter inner bound to the achievable region is derived. This is the main result of the paper, and is given in Theorem \ref{thm:MC}. It is shown that this inner bound strictly contains the CC and BT achievable regions. We provide an example where the proposed rate-distortion region is evaluated. 

The rest of the paper is organized as follows: Section \ref{Sec:Not} introduces the notations and describes the problem formulation and some of the prior works.  Section \ref{Sec:Main} presents the main results of the paper. Section \ref{sub:ach} and \ref{Sec:Ach} contain the proof of achievability of the new rate-distortion region.
\textcolor{black}{Section \ref{Sec:BOHO} provides an example where the new achievable region is shown to strictly contain the BT and CC regions.}
Section \ref{Sec:Conc} concludes the paper.

\section{Preliminaries}
\label{Sec:Not}
\subsection{Notation}
\textcolor{black}{Random variables are denoted by capital letters such as $X, U$. Sets are denoted by calligraphic letters such as $\mathcal{X}, \mathcal{U}$.} \textcolor{black}{ The expectation operator  is denoted by $\mathbb{E}(\cdot)$. The set of natural numbers and real numbers are shown by $\mathbb{N}$, and $\mathbb{R}$, respectively. 
The finite set $\{1,2,\ldots,n\}$ is denoted by $[1,n]$.}
The binary entropy function is denoted by $h_b(\cdot)$. \textcolor{black}{Convex closure is denoted by
$\mbox{cl}(\cdot)$.}
A sequence of length $n$ is denoted by $x^n$. The  $i${th} element of the vector $x^n$ is denoted by $x(i)$, \textcolor{black}{and the subsequence consisting of $x(i),x(i+1),\cdots,x(j)$ is shown by $x{[i,j]}$ for $i<j$. The term $N(a,b|x^n,y^n)$ denotes the number of occurrences of the symbol pair $(a,b)$ in the sequence pair $(x^n,y^n)$. 
} A two dimensional matrix of size $n\times m$ is denoted by $x^{n,m}$. For a binary string $x^n$, the quantity, $w_H(x^n)$, denotes its Hamming weight.  For an alphabet $\mathcal{X}$, the generalized Hamming distortion function is defined as $d_{\mathcal{X}}:\mathcal{X}\times\mathcal{X}\to \mathbb{R}^+$, where:
\[d_{\mathcal{X}}(x,\widehat{x})=
\begin{cases}
0\qquad& \text{ if } x=\widehat{x},\\
1 & \text{otherwise.}
\end{cases}
\]
 
\subsection{Problem Formulation}
\label{Sec:PS}
The two-user distributed source coding problem is depicted in Figure \ref{Fig: DSC}. 
\textcolor{black}{Consider a pair of correlated discrete memoryless sources $X_1$ and $X_2$ with alphabets $\mathcal{X}_1$ and $\mathcal{X}_2$, respectively, and joint probability distribution $P_{X_1,X_2}$. Let $\widehat{\mathcal{X}}_i$, $i= 1,2$, denotes the $i$th reconstruction alphabet.  There are two encoders and a decoder. The pair of encoders observe the sources in a distributed manner. To elaborate, the $i$th encoder observes a sequence of independent and identically distributed source variables $X_i^n$ for $i\in \{1,2\}$, generated from $P_{X_1,X_2}$, where $n$ denotes blocklength. The $i$th encoder maps the corresponding set of $n$-length sequences into a finite set of indices, and transmits the resulting index to the centralized decoder.
 The decoder produces a pair of reconstruction sequences $(\widehat{X}^n_1,\widehat{X}^n_2)$ taking values in 
 $\widehat{\mathcal{X}}_1^n \times \widehat{\mathcal{X}}_2^n$. The reconstruction is evaluated based on a pair of single-letter additive bounded distortion functions $d_i\!:\!\mathcal{X}_i\times \widehat{\mathcal{X}}_i\rightarrow \mathbb{R}^+,i=1,2$. In the following, we formally define a pair of distributed sources.}

\begin{Definition}[\textcolor{black}{\textbf{Distributed Source}}]
 \textcolor{black}{A distributed source is a tuple $(\mathcal{X}_1,\mathcal{X}_2,\widehat{\mathcal{X}}_1,\widehat{\mathcal{X}}_2,P_{X_1,X_2},d_1,d_2)$, where the finite sets $\mathcal{X}_i$ and $\widehat{\mathcal{X}}_i, i\in \{1,2\}$, are the source and reconstruction alphabets, respectively, $P_{X_1,X_2}$ is the joint  probability distribution of the source, and $d_i:\mathcal{X}_i\times\widehat{\mathcal{X}}_i\to\mathbb{R}^+, i\in \{1,2\}$ are the bounded distortion functions.} 
\end{Definition}


\begin{Remark}
For ease of notation, from now on, we restrict our attention to the case where the reconstruction alphabets are the same as the input alphabets. Hence, we often denote a distributed source by $(\mathcal{X}_1,\mathcal{X}_2,P_{X_1,X_2},d_1,d_2)$.  The results presented here can be extended to the general case in a straightforward way. 
 \end{Remark}

\begin{Definition}[\textcolor{black}{\textbf{DSC Code}}]
\label{Def:DSC}
\textcolor{black}{An $(n,\Theta_1,\Theta_2)$ code for  a given distributed source $(\mathcal{X}_1,\mathcal{X}_2,P_{X_1,X_2},d_1,d_2)$ consists of a pair of encoding functions
 \begin{align*}
 e_{i}^{(n)}:\mathcal{X}_i^n\to \{1,2,\cdots ,\Theta_i\},\quad i\in \{1,2\},
\end{align*}
and a pair of decoding functions}
\begin{align*}
\textcolor{black}{ f_i^{(n)}: \{1,2,\cdots,\Theta_1\}\times \{1,2,\cdots,\Theta_2\}\to \mathcal{X}_i^n,\quad i\in\{1,2\}.} 
\end{align*}
\end{Definition}
\textcolor{black}{In the following, we give a definition of the achievable rate-distortion (RD) region for the DSC setup.}
\begin{Definition}[\textcolor{black}{\textbf{Achievable RD Region}}]
 \textcolor{black}{For a given distributed source $(\mathcal{X}_1,\mathcal{X}_2,P_{X_1,X_2},d_1,d_2)$, 
 an RD tuple $(R_1,R_2,\Delta_1,\Delta_2)$ is said to be achievable if for all $\epsilon>0$, and all sufficiently large $n$, there exists 
 an $(n,\Theta_1,\Theta_2)$ code such that:
 \begin{align*}
\frac{1}{n} \log \Theta_i \leq R_i+\epsilon, \quad 
 \frac{1}{n}\sum_{j=1}^n \mathbb{E}\left(d_i\left(X_i(j),\widehat{X}_i(j)\right)\right)\leq \Delta_i+\epsilon, \quad i\in \{1,2\},
 \end{align*}
 where the reconstruction vectors 
 $\widehat{X}_1^n$ and $\widehat{X}_2^n$ are given by
\begin{align*}
 \widehat{X}_i^n=f_i^{(n)}\left(e_1^{(n)}(X_1^n),e_2^{(n)}(X_2^n)\right),\quad  i\in \{1,2\}.
\end{align*}
\textcolor{black}{The set of all achievable RD tuples is called the RD region and is denoted by $\mathcal{RD}^*(\mathcal{X}_1,\mathcal{X}_2,P_{X_1,X_2},d_1,d_2)$.}
\label{Def:ach}}\footnote{\textcolor{black}{We often write  $\mathcal{RD}^*(P_{X_1,X_2},d_1,d_2)$ to denote the RD region when the alphabets $\mathcal{X}_1$ and $\mathcal{X}_2$ are clear from the context.}}
\label{def:RD_Region}
\end{Definition}

 \textcolor{black}{It can be noted from the above definition that the RD region is closed and convex. We propose a new coding scheme, and derive an inner bound to the achievable RD region for the DSC setup described in Definition \ref{Def:DSC}.  The derivation includes studying the DSC in the presence of side-information (DSC-SI) which is formalized below. }

\begin{Definition}[\textcolor{black}{\textbf{Distributed Source with Side-information}}]
\textcolor{black}{A distributed source with side-information is a tuple $(\mathcal{X}_1,\mathcal{X}_2,{\mathcal{Y}_1},{\mathcal{Y}_2}, P_{X_1,X_2,{Y}_1,{Y}_2},d_1,d_2)$, where $(\mathcal{X}_1,\mathcal{X}_2, P_{X_1,X_2},d_1,d_2)$ is a distributed source, and ${Y}_1$ and ${Y}_2$ are the side-information with alphabets $\mathcal{Y}_1$ and $\mathcal{Y}_2$, respectively, distributed according to $P_{Y_1,Y_2|X_1,X_2}$. The side-information $Y_1$ and $Y_2$ are 
available at Encoder 1 and Encoder 2, respectively. The side-information pair $(Y_1,Y_2)$ is available at the decoder. }
\end{Definition}
 \begin{Definition}[\textcolor{black}{\textbf{DSC-SI Code and Achievability}}]
 \label{Def:DSCSI}
 \textcolor{black}{An $(n,\Theta_1,\Theta_2)$ code for a distributed source with side-information $(\mathcal{X}_1,\mathcal{X}_2,{\mathcal{Y}_1},{\mathcal{Y}_2},P_{X_1,X_2,Y_1,Y_2},d_1,d_2)$ consists of a pair of encoders
 \begin{align*}
 e_{i}^{(n)}:\mathcal{X}_i^n\times {\mathcal{Y}_i}^n\to \{1,2,\cdots ,\Theta_i\}, \quad i\in \{1,2\},
\end{align*}
and a pair of decoders 
\begin{align*}
 f_i^{(n)}: \{1,2,\cdots,\Theta_1\}\times \{1,2,\cdots,\Theta_2\}\times {\mathcal{Y}_1}^n\times {\mathcal{Y}_2}^n\to \mathcal{X}_i^n,\quad i\in\{1,2\}.
\end{align*}
An  RD tuple $(R_1,R_2,\Delta_1,\Delta_2)$ is said to be achievable if for all $\epsilon>0$, and all sufficiently large n, there exists an $(n,\Theta_1,\Theta_2)$ code such that:
\begin{align*}
\frac{1}{n} \log \Theta_i \leq R_i+\epsilon, \quad 
 \frac{1}{n}\sum_{j=1}^n \mathbb{E}\left(d_i\left(X_i(j),\widehat{X}_i(j)\right)\right)\leq \Delta_i+\epsilon, \quad i\in \{1,2\},
 \end{align*}
where the reconstruction vectors $\widehat{X}_1^n$ and $\widehat{X}_2^n$ are given by
\begin{align*}
 \widehat{X}_i^n=f_i^{(n)}(e_1^{(n)}(X_1^n,{Y}_1^n),e_2^{(n)}(X_2^n,{Y}_2^n),Y_1^n,Y_2^n),\quad  i\in \{1,2\}.
\end{align*}
The achievable RD region 
$\mathcal{RD}^*(P_{X_1,X_2,Y_1,Y_2},d_1,d_2)$ for the DSC-SI setup is defined as the set of all achievable RD tuples.}
\end{Definition}
\begin{Remark}
 \textcolor{black}{\textcolor{black}{An alternative formulation of the DSC-SI setup considered in prior works characterizes the problem by a tuple $(\mathcal{X}_1,\mathcal{X}_2,{\mathcal{Y}},P_{X_1,X_2,Y},d_1,d_2)$, where $Y$ is the side-information available at the decoder. It should be noted that the two formulations lead to equivalent source coding setups.} }
\end{Remark}
\subsection{Prior  Works}
\label{Sec:PR}
\textcolor{black}{
In this section, we describe relevant prior works on the DSC problem.  
We first describe the  CC achievable region which is the best known inner bound to the achievable RD region for the DSC setup. 
 \begin{Definition}[Common Component]
For two sources $X_1$ and $X_2$, the common information between them is defined as:
\begin{align}
    \label{eq:CC}
K(X_1;X_2) \triangleq  \max_{\{(f_1,f_2): V=f_1(X_1)=f_2(X_2)\}} H(V).
\end{align}
The sources are said to have a (non-trivial) common component if $K(X_1;X_2)>0$. Any non-trivial random variable $V$ which can be written as $V=f_1(X_1)=f_2(X_2)$ with probability one is called a common component of the sources.
\end{Definition}}

\begin{Theorem}[\textcolor{black}{\textbf{DSC-CC Achievable Region \cite{wagner}}}]
\textcolor{black}{For a distributed source $(\mathcal{X}_1,\mathcal{X}_2,P_{X_1,X_2},d_1,d_2)$, let $\mathcal{P}_{CC}$ denote the collection of conditional distributions $P_{W,U_1,U_2,\widehat{X}_1,\widehat{X}_2|X_1X_2}$ defined on $\mathcal{W}\times \mathcal{U}_1 \times \mathcal{U}_2 \times {\mathcal{X}_1} \times {\mathcal{X}_2}$, for some finite sets $\mathcal{W}, \mathcal{U}_1, \mathcal{U}_2 $, 
such that  (i) $W-S-(X_1,X_2)\ $ form a Markov chain, where $S$ is a common component of the sources $X_1$ and $X_2$, (ii) $U_1 - (X_1,W) - (X_2,W) - U_2$ form a Markov chain, 
 and (iii) $\widehat{X}_i=g_i(W,U_1,U_2)$, for single-letter functions $g_i, i\in\{1,2\}$. For a distribution,  $P_{W,U_1,U_2,\widehat{X}_1,\widehat{X}_2|X_1,X_2} \in \mathcal{P}_{CC}$, let $\alpha(P_{W,U_1,U_2,\widehat{X}_1,\widehat{X}_2|X_1,X_2})$ denote the set of RD tuples  $(R_1,R_2,\Delta_1,\Delta_2) \in [0,\infty)^4$ that  satisfy 
\begin{align}
R_1 &\geq I(X_1;U_1|W,U_2), \label{eq:CC1} \\
R_2 &\geq I(X_2;U_2|W,U_1), \label{eq:CC2} \\
R_1+R_2 &\geq I(X_1,X_2;U_1,U_2,W), \label{eq:CC3}\\
\Delta_i&\geq \mathbb{E}( d_i(X_i,\widehat{X}_i)),\quad  i\in\{1,2\},   \label{eq:CC4}
\end{align}
where the mutual information terms and the expectations are evaluated  with
$P_{X_1,X_2}P_{W,U_1,U_2,\widehat{X}_1,\widehat{X}_2|X_1,X_2}$. The CC rate-distortion region is defined as
\[
\mathcal{RD}_{CC}(P_{X_1,X_2},d_1,d_2) = \mbox{cl} \left(  \bigcup_{P_{W,U_1,U_2,\widehat{X}_1,\widehat{X}_2|X_1,X_2} \in
    \mathcal{P}_{CC}} \alpha(P_{W,U_1,U_2,\widehat{X}_1,\widehat{X}_2|X_1,X_2})\right).
\]
\label{Def:CC}
The CC region is achievable: }
\[\textcolor{black}{\mathcal{RD}_{CC}(P_{X_1,X_2},d_1,d_2)\subset \mathcal{RD}^*(P_{X_1,X_2},d_1,d_2).}\]
\label{thm:CC}
\end{Theorem}
\textcolor{black}{
The CC region uses the common-component of the two sources to `break' the long Markov chain $U_1 - X_1 - X_2 - U_2$ which appears in the BT achievable region. This improves upon the BT achievable region in the presence of common components among the sources.  
This is done by first quantizing the common-component at both encoders using an identical quantizer (i.e. same codebook corresponding to random variable $W$). Since both encoders have access to the common-component and they use the same quantizer,  they produce the same quantization vector $W^n$. This allows the encoders to use $W^n$ as side-information.
 }
\textcolor{black}{
The following theorem provides an inner bound to the achievable RD region for the distributed sources coding problem with side-information, which will be used in the subsequent sections.}
\begin{Theorem}[\textcolor{black}{\textbf{DSC-SI Achievable Region  \cite{varsheneya2006lossy}}}]
\textcolor{black}{
For a distributed source with side-information  $(P_{X_1,X_2,{Y}_1,{Y}_2},d_1,d_2)$, let $\mathcal{P}_{BTSI}$ denote the collection of conditional probability distributions $P_{Q,U_1,U_2,\widehat{X}_1,\widehat{X}_2|X_1,X_2,{Y}_1,{Y}_2}$ defined on $\mathcal{Q} \times \mathcal{U}_1 \times \mathcal{U}_2 \times {\mathcal{X}_1} \times {\mathcal{X}_2}$, for some finite sets $\mathcal{Q}, \mathcal{U}_1, \mathcal{U}_2$,
such that  (i) $U_1 - (X_1,{Y}_1,Q) - (X_2,{Y}_2,U_2)$ and $(U_1,X_1,{Y}_1) - (X_2,{Y}_2,Q)-U_2$ form  Markov chains, (ii) $Q$ is independent of $(X_1,X_2,{Y}_1,{Y}_2)$, and (iii) $\widehat{X}_i=g_i(U_1,U_2,{Y}_1,{Y}_2,Q)$ for single-letter functions $g_i, i\in \{1,2\}$. For a $P_{Q,U_1,U_2,\widehat{X}_1,\widehat{X}_2|X_1,X_2,{Y}_1,{Y}_2} \in \mathcal{P}_{BTSI}$, let $\alpha(P_{Q,U_1,U_2,\widehat{X}_1\widehat{X}_2|X_1,X_2,Y_1,Y_2})$ denote the set of RD tuples  $(R_1,R_2,\Delta_1,\Delta_2) \in [0,\infty)^4$ that  satisfy 
\begin{align*}
 R_1&\geq I(X_1;U_1|U_2,{Y}_1,{Y}_2,Q),\\
 R_2&\geq I(X_2;U_2|U_1,{Y}_1,{Y}_2,Q),\\
 R_1+R_2&\geq I(X_1,X_2;U_1,U_2|{Y}_1,{Y}_2,Q), \\
 \Delta_i &\geq  \mathbb{E}(d_i(X_i,\widehat{X}_i)), \ i \in \{1,2\}  ,
\end{align*}
 where the mutual information terms and expectations are evaluated  with
$P_{X_1,X_2,Y_1,Y_2}P_{Q,U_1,U_2,\widehat{X}_1,\widehat{X}_2|Y_1,Y_2,X_1,X_2}$. The Berger-Tung with Side-information (BTSI) rate-distortion region is defined as
\[
\mathcal{RD}_{BTSI}(P_{X_1,X_2,Y_1,Y_2},d_1,d_2) = \mbox{cl} \left( \bigcup_{P_{Q,U_1,U_2,\widehat{X}_1,\widehat{X}_2|Y_1Y_2X_1X_2} \in
    \mathcal{P}_{BTSI}} \alpha(P_{Q,U_1,U_2,\widehat{X}_1\widehat{X}_2|Y_1,Y_2,X_1,X_2}) \right).
\]
The BTSI rate-distortion region is achievable: 
\[\mathcal{RD}_{BTSI}(P_{X_1,X_2,Y_1,Y_2},d_1,d_2)\subseteq \mathcal{RD}^*(P_{X_1,X_2,Y_1,Y_2},d_1,d_2).\]
\label{thm:BTSI}}
\end{Theorem}

\section{The New Achievable Rate-distortion Region}
\label{Sec:Main}
\textcolor{black}{
In this section, we provide a new coding scheme called the \textit{Finite Length Matrix Coding} (FLMC) scheme, and derive 
a new inner bound, given in 
Theorem \ref{thm:MC}, 
to the achievable RD region for the general DSC problem, which is the main result of the paper. Toward describing the scheme, we define \textit{$\epsilon$-correlated components} of the sources as follows:}
\begin{Definition}[\textcolor{black}{\textbf{$\epsilon$-Correlated Components}}]
\textcolor{black}{A pair of random variables $\left(S_{1},S_{2}\right)$,  each taking values from a finite set $\mathcal{S}$, is said to be a pair of $\epsilon-$correlated components of the source pair $(X_1,X_2)$, if there exists an $\epsilon\in [0,1]$ and functions $f_i: \mathcal{X}_i\to \mathcal{S}, i\in \{1,2\}$, such that:}
\[
\textcolor{black}{S_{1}=f_1(X_1), \ \ S_{2}=f_2(X_2), \ \ P(S_{1}=S_{2})\geq 1-\epsilon.}
\] 
\label{Def:eps}
\end{Definition}

\begin{Remark}
A CC is a special instance of $\epsilon$-correlated components of the two sources. To elaborate, let $S=f_i(X_i), i\in \{1,2\}$ be a common component of the source pair $(X_1,X_2)$. The pair $(f_1(X_1),f_2(X_2))$ is  a pair of  $\epsilon$-correlated components of $(X_1,X_2)$ for every $ \epsilon\in [0,1]$. Conversely,
if $(S_1,S_2)$ is a pair of $0$-correlated components of $(X_1,X_2)$, then $S=S_1=S_2$ is a common component of the sources. 
\end{Remark}


\textcolor{black}{
Given an $\epsilon\in [0,1]$, and a pair of $\epsilon$-correlated components of the sources, in the FLMC scheme, 
we construct a  
pair of encoder mappings $e_i^{(nm)}, i \in \{1,2\}$, and decoder mappings $f_i^{(nm)}, i \in \{1,2\}$, which operate on sequences of length $nm$, where $n$ is a fixed number and $m$ tends to infinity. The number $n$ is chosen to satisfy certain conditions with respect to $\epsilon$, as described in Theorem \ref{thm:MC}. The encoding and decoding mappings are based on  two layers of codes. To construct the first layer, we generate random quantizers which operate on sequences of length $n$. We quantize the $\epsilon$-correlated components of the sources by applying these randomly generated quantizers repeatedly on $m$ sub-blocks, each of length $n$. The second layer uses the conventional large blocklength random codes of length $m$ to quantize the sources by applying these quantizers repeatedly on $n$ sub-blocks, each of length $m$. 
 These second layer codes are similar to the codes which are used in the BT and CC schemes. The two coding layers are interleaved with each other using a new interleaving method which is described in the proof of Theorem \ref{thm:MC}. We  characterize the asymptotic performance of this scheme as $m$ tends to infinity using single-letter information quantities.  
 We call the resulting achievable RD region the Finite Length Matrix Coding rate-distortion region (FLMC-RD).}


\begin{Theorem} [\textcolor{black}{\textbf{FLMC Achievable RD}}]
\textcolor{black}{Consider a distributed source $(\mathcal{X}_1,\mathcal{X}_2,P_{X_1,X_2},d_1,d_2)$. Let $(S_1, S_2)$ be a pair of $\epsilon$-correlated components of $X_1$ and $X_2$ with alphabet $\mathcal{S}$, where $\epsilon \in [0,1]$. Let $\mathcal{P}_{FLMC}$ denote the collection of conditional probability distributions $P_{W,U_1,U_2,\widehat{X}_1,\widehat{X}_2|X_1,X_2}$ defined on $\mathcal{W}\times \mathcal{U}_1 \times \mathcal{U}_2 \times {\mathcal{X}_1} \times {\mathcal{X}_2}$, for some finite sets $\mathcal{W},\mathcal{U}_1$ and $\mathcal{U}_2$,
such that  (i) $|\mathcal{W}|\leq |\mathcal{S}|+1$, and $|\mathcal{U}_i|\leq |\mathcal{X}_i|\cdot|\mathcal{W}|+1$,  for $i\in \{1,2\}$,
(ii) $W-S_1-(X_1,X_2)\ $ form a Markov chain, (iii) $U_1 - (X_1,W) - (X_2,W) - U_2$ form a Markov chain, and (iv) $\widehat{X}_i=g_i(W,U_1,U_2), i\in\{1,2\}$, for single-letter functions
$g_i$.} 
\textcolor{black}{Define a function $\delta_k\triangleq1-(1-\epsilon)^k, \mbox{ for all } k\in \mathbb{N}$, and a set $\mathcal{B}(\epsilon)\triangleq\{n\in \mathbb{N}: \left(\frac{\sigma}{\sigma'}\right)^2  \leq n \leq  \frac{\log(2\epsilon)}{\log(1-\epsilon)}\}$. For a 
pair of $\epsilon$-correlated components $(S_1,S_2)$ of the sources $X_1,X_2$, a 
distribution $P_{W,U_1,U_2,\widehat{X}_1,\widehat{X}_2|X_1,X_2} \in \mathcal{P}_{FLMC}$, a parameter $n\in \mathcal{B}(\epsilon)$, and
a parameter $\tau \in (\frac{1}{\sqrt{n}}\sigma,\sigma')$,
let $\alpha_{FLMC}((S_1,S_2),P_{W,U_1,U_2,\widehat{X}_1,\widehat{X}_2|X_1,X_2},\tau,n)$ denote the set of rate-distortion tuples $(R_1,R_2,\Delta_1,\Delta_2) \in [0,\infty)^4$ that  satisfy
\begin{align}
\label{eq:01}R_1&\geq I(X_1;U_1|U_2,W) + E_{n,\delta_n} +\Gamma_{n,\delta_n}(\tau)+\Lambda_n(\tau),\\
\label{eq:02}R_2&\geq I(X_2;U_2|U_1,W) +E_{n,\delta_n}+\Gamma_{n,\delta_n}(\tau),\\
\label{eq:03}R_1+R_2&\geq I(X_1,X_2;U_1,U_2,W)+E_{n,\delta_n}+\Gamma_{n,\delta_n}(\tau)+\Lambda_n(\tau)+\theta_n(\tau),\\
\label{eq:04}\Delta_i&\geq \mathbb{E}\{d_i\big(X_i,\widehat{X}_i \big)\}+ 2\left(1-p(\tau)+p(\tau)\delta_n+\epsilon\right) d_{i,max}
, \ i\in \{1,2\},
\end{align}}
\noindent\textcolor{black}{ where:
\\i) $\sigma\triangleq \max(|\mathcal{S}|,|\mathcal{W}|)\sqrt{\ln{\frac{8\max(|\mathcal{S}|,|\mathcal{W}|)}{\min_{a,b\in\mathcal{S}\times\mathcal{W}}P_{S_1,W}(a,b)}}}
$ and $\sigma'\triangleq \frac{|\mathcal{S}||\mathcal{W}|}{2(|\mathcal{S}|+|\mathcal{W}|)}\min_{a,b\in\mathcal{S}\times\mathcal{W}}P_{S_1,W}(a,b)$,
\\ii) $p(\tau)
\triangleq
\frac{\min_{a,b\in\mathcal{S}\times\mathcal{W}}P_{S_1,W}(a,b)}{\min_{a,b\in\mathcal{S}\times\mathcal{W}}P_{S_1,W}(a,b)+4|\mathcal{S}|e^{-2n\frac{{\tau}^2}{|\mathcal{S}|^2}}+{{\tau}}\left(\frac{1}{|\mathcal{S}|}+\frac{1}{\mathcal{|W|}}\right)}$,
\\iii) { $\theta_n(\tau)$ is the finite blocklength random coding rate-loss due to quantizing the source $S_1$ using the test channel $P_{W|S_1}$, given by}
         \begin{align}
     \nonumber
     \theta_n(\tau)&\triangleq  \frac{1}{n}\log\left(2n\frac{\tau^2}{|\mathcal{S}|^2}- \ln{2|\mathcal{S}|}\right)
     \\
     &\qquad+\tau\left(\frac{1}{|\mathcal{W}|}\sum_{s\in \mathcal{S}}\sum_{w\in \mathcal{W}}\log{\frac{1}{P_{W|S_1}(w|s)}}
     +\frac{1}{|\mathcal{S}|}\sum_{s\in\mathcal{S}}H(W|S_1=s)+\frac{(|\mathcal{S}|+|\mathcal{W}|)}{|\mathcal{S}|}\sum_{s\in\mathcal{S}}\log{\frac{1}{P_{S_1}(s)}}\right)
     +\frac{1}{n}(|\mathcal{S}|+1),\label{eq:theta}
 \end{align}
\\iv) $\Gamma_{n,\delta_n}(\tau)\triangleq 8h_b({x})+8x\log{(|\mathcal{X}_1\times \mathcal{X}_2\times \mathcal{U}_1\times \mathcal{U}_2\times \mathcal{W}|)}$, where $x=\left(1-p(\tau)+p(\tau)\delta_n+\epsilon\right)$,
\\v) $\Lambda_n(\tau)\triangleq
{h_b\left(p(\tau)\right)+(1-p(\tau))\log{|\mathcal{W}|}}$, \\vi) $E_{n,\delta_n}\triangleq \frac{1}{n}{h_b(\delta_n)}+\delta_n \log|\mathcal{W}|$, \\
vii) $d_{i,max}=max_{x_i,\hat{x}_i}d_i(x_i,\hat{x}_i), i\in \{1,2\}$, 
\\and the mutual information terms and expectations are evaluated with the distribution $P_{X_1X_2}P_{WU_1U_2\widehat{X}_1\widehat{X}_2|X_1X_2}$. Define the FLMC RD region as:
\[
\mathcal{RD}_{FLMC}(P_{X_1,X_2},d_1,d_2)\triangleq \mbox{cl}\left(\bigcup_{(S_1,S_2)}
\bigcup_{P_{W,U_1,U_2,\widehat{X}_1,\widehat{X}_2|X_1,X_2}\in \mathcal{P}_{FLMC}}
\bigcup_{n\in \mathcal{B}(\epsilon)}
\bigcup_{\tau\in (\frac{1}{\sqrt{n}}\sigma,\sigma')}
\alpha_{FLMC}((S_1,S_2),P_{W,U_1,U_2,\widehat{X}_1,\widehat{X}_2|X_1,X_2},\tau,n)\right).
\]}
 The FLMC region is achievable, i.e. $\mathcal{RD}_{FLMC}(P_{X_1,X_2},d_1,d_2)\subset \mathcal{RD}^*(P_{X_1,X_2},d_1,d_2)$.
\label{thm:MC}
\end{Theorem}

\begin{Remark}
\textcolor{black}{The cardinality bounds $|\mathcal{W}|\leq |\mathcal{S}|+1$, and $|\mathcal{U}_i|\leq |\mathcal{X}_i|\cdot|\mathcal{W}|+1, i\in \{1,2\}$ in the theorem are derived using standard arguments based on the Caratheodory Theorem (e.g. \cite{rockafellar1970convex} Theorem 17.1), and by noting that the additional $\log{|\mathcal{W}|}$ term in equations \eqref{eq:01}-\eqref{eq:03} decreases as $|\mathcal{W}|$ is decreased.}
\end{Remark}

\begin{Remark}
\textcolor{black}{
Consider first the case $\epsilon>0$. If $\epsilon>\frac{1}{3}$, we have $\mathcal{B}(\epsilon)=\phi$. Otherwise, if $0<\epsilon<\frac{1}{3}$ the set $\mathcal{B}(\epsilon)$ has a finite number of elements. Hence, the region is computable in this case. For the case when $\epsilon=0$, the $\epsilon$-correlated components $S_1,S_2$ are in fact common components, i.e. $S_1=S_2$ with probability one. Furthermore, the set $\mathcal{B}(\epsilon)$ is infinite, and $\delta_n=0, E_{n,\delta_n}=0, \forall n \in \mathbb{N}$. Next, we note that as $n\to \infty$ we have $\frac{1}{\sqrt{n}}\sigma\to 0$. So, $\theta_n(\tau)$ can be made arbitrarily small as $n\to \infty$ using a suitable choice of $\tau\in (\frac{1}{\sqrt{n}}\sigma,\sigma')$, for example $\tau= n^{-\frac{1}{2}+\alpha}, 0<\alpha<\frac{1}{2}$. Similarly, $\Gamma_{n,\delta_n}(\tau)$ and $\Lambda_n(\tau)$ can be made arbitrarily small. Hence, the terms on the right hand side of the inequalities \eqref{eq:01}-\eqref{eq:04} which depend on $n$ and $\tau$ can be made  arbitrarily small by choosing  $n$ sufficiently large and $\tau$ appropriately. As a result,   
the largest achievable RD region is obtained  when $n\to \infty$ which coincides with the CC region. Hence, the region is computable in this case as well.}
\end{Remark}

\begin{Remark}
 \textcolor{black}{Note that since the alphabets $\mathcal{X}_1$ and $\mathcal{X}_2$ are finite, there are a finite number of choices for the components $S_1$ and $S_2$ which are functions of $X_1$ and $X_2$, respectively. Hence, to compute the achievable region, it is possible to compute the RD region for each choice of $(S_1,S_2)$, and then take the convex hull of the union of these RD regions using the Caratheodory theorem.}
\end{Remark}


The achievability of the FLMC region is proved in Section \ref{Sec:Ach}. The following proposition proves that the FLMC region contains the CC region. 
\begin{Proposition} 
\label{Prop:Subset}
\textcolor{black}{For any distributed source $(P_{X_1,X_2},d_1,d_2)$, the FLMC region contains the CC region: }
\[
\textcolor{black}{\mathcal{RD}_{CC}(P_{X_1,X_2},d_1,d_2)\subseteq \mathcal{RD}_{FLMC}(P_{X_1,X_2},d_1,d_2).}\]
\end{Proposition}
\textit{Proof.}
\textcolor{black}{Please see Appendix \ref{App:Subset}.}

\textcolor{black}{ Although, for $\epsilon>0$, the parameter $n$ is finite and is restricted to the set $\mathcal{B}(\epsilon)$,  to gain  better understanding,  we make the following observations regarding the asymptotic behavior of the $n$-dependent terms in the FLMC RD region as follows. 
First, the term $\theta_n(\tau)$ in the RD region provided in Theorem \ref{thm:MC} signifies the rate-loss due to the application of quantizers of constant finite blocklength in the first code layer in the proposed scheme. This rate-loss diminishes as the codebook blocklength $n$ in the first layer is increased asymptotically, and $\tau$ is chosen suitably, due to concentration of measure (i.e. $\theta_n(\tau)\to 0$ as $n\to \infty$ and $\tau \to 0$ such that $n\tau^2 \to \infty$).} \textcolor{black}{Second, the term $\Lambda_n(\tau)$ signifies the rate-loss due to the error in estimating the empirical distribution of the output of the first layer of quantizers. As the blocklength of the first layer of quantizers is increased, this estimation error approaches zero asymptotically due to the law of large numbers (i.e. $\Lambda_n(\tau)\to 0$ as $n\to \infty$ and $\tau \to 0$ such that $n\tau^2 \to \infty$.).}   
\textcolor{black}{ Third, for $\epsilon>0$, the term $E_{n,\delta_n}$ in the RD region provided in Theorem \ref{thm:MC} signifies the loss due to the decrease in  correlation among the quantizations of the highly correlated components of the sources \cite{ComInf2, shirani2019sub}. This rate-loss increases as the codebook blocklength $n$ in the first layer is increased asymptotically (i.e. $E_{n,\delta_n}\uparrow \log{|\mathcal{W}|}$ as $n\to \infty$). }
\textcolor{black}{ Fourth, $\Gamma_{n,\delta_n}(\tau)$ signifies the rate-loss due to approximation of quantizations of highly correlated components by that of a common component at the distributed encoders. This expression is not defined for asymptotically large $n$ when $\epsilon>0$. The range of $n$ given as $\mathcal{B}(\epsilon)$ and the range of $\tau$ given as $(\frac{1}{\sqrt{n}}\sigma,\sigma')$ 
is chosen such that $p(\tau)>\frac{1}{2}$, consequently $x= (1-p(\tau)+p(\tau)\delta_n+\epsilon)$ is less than one, and hence $\Gamma_{n,\delta_n}(\tau)$ is well-defined.}

\textcolor{black}{  It can be observed that there is a trade-off between the rate-loss given by $\theta_n(\tau)$ and  $\Lambda_n(\tau)$ due to the application of the first layer of codes with fixed blocklength, and the cost given by $E_{n,\delta_n}$ incurred due to the weakening of the correlation between the quantizations of the sources. 
}

    The FLMC-RD region is not symmetric with respect to the two encoders. A larger achievable region can be characterized by swapping the indices for encoders $1$ and $2$ in the theorem and taking the convex hull of the union of the two resulting regions. 

The proof of Theorem \ref{thm:MC} involves two steps. First, we prove that another larger region called the matrix coding multi-letter region (MCML) is achievable. The MCML region is not computable, and is characterized using multi-letter distributions on the input variables. In the second step, we extract a single-letter FLMC region as a 
subset of the MCML region. Since the latter is achievable, so is the former. 

\section{The MCML Achievable Region}
\label{sub:ach}
\textcolor{black}{We first provide an outline of the  MCML strategy.} In the first step, we quantize distributively the $\epsilon$-correlated components $(S_1,S_2)$ of the sources $(X_1,X_2)$ using identical quantizers. The quantizer is designed based on a specific multi-letter distortion constraint.
To elaborate, let $W_1^n$ be the sequence resulting from the application of the quantizer on  the vector $S_1^n$. Roughly speaking, the multi-letter distortion constraint requires the pair $(S_1^n,W_1^n)$  to have its empirical distribution close to a predetermined joint distribution $P_{S_1,W}$. \textcolor{black}{Note that the distortion criteria only imposes restrictions on the quantization $W^n_1$ of the sequence $S^n_1$,  and does not depend on the sequence $S_2^n$.
The second encoder estimates the first encoder's output by applying the same quantizer to $S_2^n$ to obtain $W_2^n$.} The encoders use a multi-letter Slepian-Wolf coding strategy to transmit the pair $(W_1^n,W_2^n)$. \textcolor{black}{The rest of the scheme is similar to the scheme used for the DSC-SI scenario, where $(W_1^n,W_2^n)$ is treated as distributed side-information.} Let $\mathcal{RD}_{MCML}(P_{X_1,X_2},d_1,d_2)$ denote the achievable  RD region for the MCML scheme. The following describes this achievable RD region.

\begin{Definition}
\textcolor{black}{Consider a distributed source $(\mathcal{X}_1,\mathcal{X}_2,P_{X_1,X_2},d_1,d_2)$. Let $(S_1, S_2)$ be a pair of $\epsilon$-correlated components of $X_1$ and $X_2$ with alphabet $\mathcal{S}$. 
Define $\mathcal{P}_{MCML}$ as the set of all tuples
of conditional distributions and functions $(P_{W|S_1},P_{U_1|X_1,W},$ $P_{U_2|X_2,W}, g_1, g_2)$, where the distributions are defined on  $\mathcal{W}\times \mathcal{U}_1\times \mathcal{U}_2$, $g_i:\mathcal{W}\times{\mathcal{W}}\times\mathcal{U}_1\times \mathcal{U}_2\to \mathcal{X}_i, i\in \{1,2\}$, and $\mathcal{W},\mathcal{U}_1$ and 
$\mathcal{U}_2$ are finite sets.}

\end{Definition}

\begin{Theorem} 
\textcolor{black}{Let $n\in \mathcal{B}(\epsilon)$, and $\tau\in (\frac{1}{\sqrt{n}}\sigma,\sigma')$. Consider a tuple $(P_{W|S_1},P_{U_1|X_1,W},P_{U_2|X_2,W},$ $g_1, g_2)\in \mathcal{P}_{MCML}$. Then, there exists a conditional probability distribution
$P'_{W'_1,W_2,U_1,U_2|X_1,X_2}$ such that:
\\1) $P'_{W'_1,U_1|X_1,X_2}(w_1,u_1|x_1,x_2)=P_{W|S_1}(w_1|s_1)P_{U_1|X_1,W}(u_1|x_1,w_1), \forall x_1,x_2,w_1,u_1\in \mathcal{X}_1\times \mathcal{X}_2\times \mathcal{W}\times \mathcal{U}_1,$ where $s_1=f_1(x_1)$,
\\2) $P'_{U_2|X_2,W_2}(u_2|x_2,w_2)=P_{U_2|X_2,W}(u_2|x_2,w_2), \forall u_2,x_2,w_2\in \mathcal{U}_2\times \mathcal{X}_2\times \mathcal{W}$,
\\3) $P'(W'_1=W_2,S_1=S_2)\geq p(\tau)(1-\delta_n)$,
\\4) $U_1 - (X_1,W'_1) - (X_2,W_2) - U_2$,
\\5) $(W'_1,W_2) - (S_1,S_2) - (X_1,X_2)$.}
\\\textcolor{black}{
Let $\alpha_{MCML}(P_{W|S_1},P_{U_1|X_1,W},P_{U_2|X_2,W}, g_1, g_2,\tau,n)$ denote\footnote{The dependence of $\alpha_{MCML}$ on $(S_1,S_2)$ is not made explicit.} the set of rate-distortion tuples 
$(R_1,R_2,\Delta_1,\Delta_2) \in [0,\infty)^4$  that satisfy the following bounds:
\begin{align}
\label{eq:prembounds1}R_1 &\geq I(X_1;U_1|U_2,W'_1,W_2) + E_{n,\delta_n}+\Lambda_n(\tau),\\
\label{eq:prembounds2}R_2&\geq I(X_2;U_2|U_1,W'_1,W_2) +E_{n,\delta_n},\\
\label{eq:prembounds3}R_1+R_2&\geq I(W'_1;S_1)+I(X_1;U_1|W'_1,W_2)+ I(X_2;U_2|W'_1,W_2)-I(U_1;U_2|W'_1,W_2)+E_{n,\delta_n} +\Lambda_n(\tau) +\theta_n(\tau),\\
\Delta_i&\geq \mathbb{E}\{d_i\big(X_i,g_i(U_1,U_2,W'_1,W_2)\big)\}, 
\ i \in \{1,2\},
\label{eq:prembounds4}
\end{align}
 where the mutual information terms and the expectations are evaluated
 according to the distribution $P_{X_1,X_2}$ $P'_{W'_1,W_2,U_1,U_2|X_1,X_2}.$
 Define the MCML RD region as 
 \begin{align*}
&\mathcal{RD}_{MCML}(P_{X_1,X_2},d_1,d_2) \\
&\triangleq \mbox{cl}\left(\bigcup_{(S_1,S_2)}
\bigcup_{(P_{W|S_1},P_{U_1|X_1W},P_{U_2|X_2W},g_1,g_2) \in \mathcal{P}_{MCML}}
\bigcup_{n \in \mathcal{B}(\epsilon)}
\bigcup_{\tau\in (\frac{1}{\sqrt{n}}\sigma,\sigma')}
\alpha_{MCML}(P_{W|S_1},P_{U_1|X_1,W},P_{U_2|X_2,W}, g_1, g_2,\tau,n)\right).
\end{align*}
The MCML region is achievable, i.e. $\mathcal{RD}_{MCML}(P_{X_1,X_2},d_1,d_2)\subset \mathcal{RD}^*(P_{X_1,X_2},d_1,d_2)$.
}
\label{thm: MCDSC22}
\end{Theorem}


\begin{proof} The proof is given in 6 steps.
\\\textbf{Step 0.}
\textcolor{black}{ Let $(S_1,S_2)$ be a pair of $\epsilon$-correlated components of the distributed source $(X_1,X_2)$. Fix a tuple
$(P_{W|S_1},P_{U_1|X_1,W},$ $P_{U_2|X_2,W}, g_1, g_2) 
\in \mathcal{P}_{MCML}$, an $n\in \mathcal{B}(\epsilon)$, and a $\tau \in (\frac{1}{\sqrt{n}}\sigma,\sigma')$.  The distribution $P'_{W'_1,W_2,U_1,U_2|X_1,X_2}$, whose existence will be shown in the following, depends on the choice of the finite blocklength quantizer in the first quantization step. This quantizer is used to quantize the correlated components $(S_1,S_2)$. 
The finite blocklength quantizer is designed with respect to the multi-letter distortion criterion described below.} 
\begin{Definition}
\textcolor{black}{
 Let $\mathcal{W}$ be an arbitrary finite set. For a given conditional distribution $P_{W|S_1}$ defined on alphabet $\mathcal{W}$, define the following n-letter distortion function
 \[d^n_{S_1,W}(s_1^n,w_1^n)=
\max_{a,b\in \mathcal{S}\times \mathcal{W}} |\mathbb{P}_{s_1^n,w_1^n}(a,b)-  P_{S_1,W}(a,b)|,  
\quad  (s_1^n,w_1^n) \in \mathcal{S}^n\times\mathcal{W}^n\] where $P_{S_1,W}=P_{S_1}P_{W|S_1}$, and $\mathbb{P}_{s_1^n,w_1^n}$ is the joint type of $(s_1^n,w_1^n)$:}
 \begin{align*}
\textcolor{black}{     \mathbb{P}_{s_1^n,w_1^n}(a,b)= \frac{1}{n}
N(a,b|s_1^n,w_1^n), \ (a,b)\in \mathcal{S}\times \mathcal{W},}
 \end{align*}
 \textcolor{black}
{ where $N(a,b|s_1^n,w_1^n)\triangleq \sum_{j=1}^n\mathbbm{1}(s_{1}(j)=a, w_{1}(j)=b), (a,b)\in \mathcal{S}\times \mathcal{W}$.}
 \label{def:dist}
\end{Definition}
\textcolor{black}{The motivation behind defining the distortion $d^n_{S_1,W}(\cdot,\cdot)$ as above is as follows. We are interested in constructing quantizers for which the empirical distribution of the input $S^n_1$ and output $W_1^n$ is close to $P_{S_1,W}$. Consequently, we want to find quantizers which yield small $d^n_{S_1,W}(\cdot,\cdot)$.}
The following defines an n-length fixed blocklength quantizer for the distortion function $d^n_{S_1,W}$, and joint distribution $P_{S_1,W}$ as described above:
\begin{Definition}\textcolor{black}{
For $\phi, \phi'\in [0,1]$, an $(n,M,\phi,\phi')$ quantizer for the source $(\mathcal{S},\mathcal{W},P_{S_1}, d^n_{S_1,W})$ is a function $Q_n:\mathcal{S}^n_1\to \mathcal{W}^n$ such that $|Im(Q_n)|=M$, and $P(d^n_{S_1,W}(S^n,Q_n(S^n))>\phi)\leq \phi'$. }  
\label{def:Q}
\end{Definition}
An optimal n-length quantizer is formally defined below:
\begin{Definition}
\textcolor{black}{
 Let $M^*_{n,\phi,\phi'}$ be the minimum $M$ such that an $(n,M,\phi,\phi')$ quantizer exists for the source $(\mathcal{S},\mathcal{W},P_{S_1}, d^n_{S_1,W})$.
 The corresponding quantizer is denoted by $Q^*_{n,\phi,\phi'}$.
 The rate of the quantizer is defined as $R^*_{n,\phi,\phi'}= \frac{1}{n}\log{M^*_{n,\phi,\phi'}}$.}
 \label{def:optQ}
\end{Definition}

\begin{Proposition}
\textcolor{black}{ For a source $(\mathcal{S},\mathcal{W},P_{S_1}, d^n_{S_1,W})$, an $n> \left(\frac{\sigma}{{\sigma'}}\right)^2$, and
$\tau\in (\frac{1}{\sqrt{n}}\sigma,\sigma')$, the following holds: 
\begin{align*}
    &R^*_{n,\phi,\phi'}\leq I(W;S_1)+\theta_n(\tau),
    \\
     &\phi= {\tau}\left(\frac{1}{|\mathcal{S}|}+\frac{1}{\mathcal{|W|}}\right),\quad 
     \phi'=
     4|\mathcal{S}|e^{-2n\frac{{\tau}^2}{|\mathcal{S}|^2}},
\end{align*}
and $\theta_n$ is defined in \eqref{eq:theta}.
}
\label{prop:fin}
\end{Proposition}

\begin{proof}
\textcolor{black}{Please see Appendix \ref{App:prop:fin}.}
\end{proof}
\textcolor{black}{
The first encoder uses the optimal quantizer $Q^*_{n,\phi,\phi'}$ to quantize $S_1^n$ to $W_1^n$, where $\phi$ and $\phi'$ are chosen as described in the Proposition \ref{prop:fin}.  
Let $W_2^n=Q_{n,\phi,\phi'}(S_2^n)$. Note that $W_2^n$ is the second encoder\rq{}s \lq\lq{}estimate\rq\rq{} of $W_1^n$. Define:
\begin{align*}
P'_{X_1^n,X_2^n,W_1^n,W_2^n}(x_1^n,x_2^n,w_1^n,w_2^n)\triangleq&
{P_{X_1^n,X_2^n}(x_1^n,x_2^n)P'_{W_1^n,W_2^n|S_1^n,S_2^n}}(w_1^n,w_2^n|s_1^n,s_2^n),\end{align*}
for all $x_1^n,x_2^n, w_1^n,w_2^n\in \mathcal{X}_1^n\times \mathcal{X}_2^n\times \mathcal{W}^n\times \mathcal{W}^n$, $s_i^n=\big(f_i(x_i(1)), f_i(x_i(2)),\cdots, f_i(x_i(n))\big), i\in \{1,2\}$,  where 
\[
P'_{W_1^n,W_2^n|S_1^n,S_2^n}(w_1^n,w_2^n|s_1^n,s_2^n)= 
\mathbbm{1}(w_1^n=Q_{n,\phi,\phi'}^*(s_1^n),
w_2^n=Q_{n,\phi,\phi'}^*(s_2^n)),
\]
i.e., the distribution induced by applying the quantizer on $S_1^n$ and $S_2^n$.}

Let $m \in \mathbb{N}$.
Note that $n$ is the length of the finite blocklength quantizers used in the first layer of coding, and $m$ is the length of the quantizers in the second layer which is taken to be asymptotically large. \textcolor{black}{ Encoder $i$ receives the source vector $X_i^{n,m}$, where $i\in \{1,2\}$.} Let $\mathcal{S}_n$ be the set of all permutations $\pi:[1,n]\to[1,n] $. \textcolor{black}{Prior to the start of communication, select $m$ permutations $\pi_1,\pi_2,\cdots,\pi_m\in \mathcal{S}_n$ randomly, independently and uniformly.} These permutations are given to both encoders as well as the decoder.
\\\textbf{Step 1.} Transmitting $\left(W_1^{n,m},W_2^{n,m}\right)$:
\\In this step, the encoders transmit the quantizations of the $\epsilon$-correlated components of the sources. \textcolor{black}{These components are quantized using the optimal quantizer whose blocklength is $n$. 
Let ${W_i}([1,n],j)\triangleq Q^*_{n,\phi,\phi'}({S_i}([1,n],j))$, $i\in\{1,2\}, \ j\in [1,m]$.  The encoders transmit $(W_1^{n,m},W_2^{n,m})$ losslessly.} \textcolor{black}{More precisely, let the `\textit{super variable}' $\overline{W}_{i}$ be defined on the alphabet $\mathcal{W}^n$, $i\in \{1,2\}$. Define the sequence of \textit{super variables $\overline{W}_i^m=(W_{i}([1,n],1),W_{i}([1,n],2),$ $\cdots,W_{i}([1,n],m))$.} Consider the problem of compressing the distributed sources $\overline{W}_1$ and $\overline{W}_2$, with distribution  $P'_{\overline{W}_1,\overline{W}_2}$ shown in Figure \ref{fig:lossless}, where  $P'_{\overline{W}_1,\overline{W}_2}$ is the joint distribution of $\overline{W}_1,\overline{W}_2$ resulting from applying $Q^*_{n,\phi,\phi'}$.
From the Slepian-Wolf theorem for lossless distributed source coding \cite{SW}, the rate pair $(nr_1,nr_2)$ is achievable with block-error probability approaching zero asymptotically as $m\to \infty$, where }
\begin{align*}
 &nr_1\geq H(\overline{W}_1|\overline{W}_2),\qquad nr_2\geq  H(\overline{W}_2|\overline{W}_1),\qquad n(r_1+r_2)\geq H( \overline{W}_1,\overline{W}_2).
%
\end{align*}
\begin{figure}
\includegraphics*[draft=false,scale=.8]{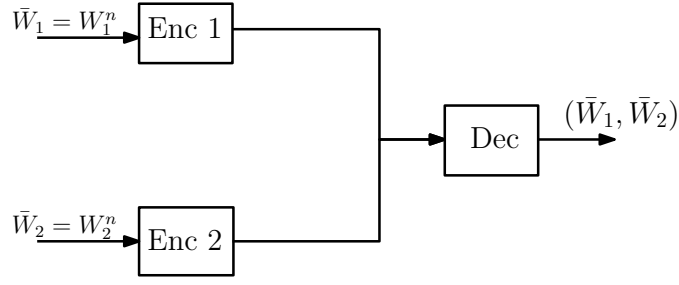}
\centering 
\caption{The encoders transmit $(\overline{W}_1,\overline{W}_2)$ losslessly.}
\label{fig:lossless}
\end{figure}
Let $\mathcal{E}_i$ be the event that $S_1([1,n],i)\neq S_2([1,n],i), i\in[1,m]$. Also, let $\mathcal{F}_i$ be the event that $W_1([1,n],i)\neq W_2([1,n],i), i\in[1,m]$. Then:
\begin{align*}
 P'_{W_1^n,W_2^n}(\mathcal{F}_i)\leq P_{S_1^n,S_2^n}(\mathcal{E}_i)=1-P_{S_1^n,S_2^n}(\mathcal{E}_i^c)=1-P_{S_1^n,S_2^n}(S_1^n=S_2^n)\leq 1-(1-\epsilon)^n=\delta_n.
\end{align*}
Define $G=\mathbbm{1}(\overline{W}_1=\overline{W}_2)$, then,
\begin{align*}
H(\overline{W}_2|\overline{W}_1)&=H(\overline{W}_2,G|\overline{W}_1)=H(G|\overline{W}_1)+H(\overline{W}_2|\overline{W}_1,G)
\\&\leq H(G)+P(G=0)H(\overline{W}_2|\overline{W}_1,G=0)+P(G=1)H(\overline{W}_2|\overline{W}_1,G=1)\\
&\leq h_b(1-(1-\epsilon)^n)+(1-(1-\epsilon)^n)H(\overline{W}_2|\overline{W}_1,G=0)\\
&\leq \textcolor{black}{h_b(1-(1-\epsilon)^n)+(1-(1-\epsilon)^n)\log{|\mathcal{W}^n|}=nE_{n,\delta_n},}
\end{align*}
\textcolor{black}{
Also, from Proposition \ref{prop:fin}: }
\begin{align*}
\textcolor{black}{ H(\overline{W}_1)= H(W_1^n){\leq}  nI(W;S_1)+n\theta_n(\tau).}
\end{align*}
   \textcolor{black}{ So, there exists a sequence of encoding functions $e_{i,m}:{\mathcal{W}}^{n, m}\to \{0,1\}^{l_{1,i}}, i\in \{1,2\},m\in \mathbb{N}$, and decoding functions $f_{m}:\{0,1\}^{l_{1,1}}\times\{0,1\}^{l_{1,2}}\to {\mathcal{W}}^{n, m}\times {\mathcal{W}}^{n, m}$ such that $\frac{l_{1,i}}{m}\leq nr_i+\epsilon_m, i\in \{1,2\}$, where $\epsilon_m\to 0$ as $m\to \infty$, and $P'\left(f_m\left(e_{1,m}(\overline{W}_1^m),e_{2,m}(\overline{W}_2^m)\right)\neq (\overline{W}_1^m,\overline{W}_2^m)\right)\to 0$ as $m\to \infty$ for any pair $(r_1,r_2)$ satisfying the following:}
\begin{align*}
 r_1\geq E_{n,\delta_n},\qquad r_2\geq E_{n,\delta_n}, \qquad \textcolor{black}{r_1+r_2\geq I(W;S_1)+E_{n,\delta_n}+\theta_n(\tau).}
\end{align*}

\textcolor{black}{In this step, Encoder $i$ transmits $e_{i,m}(\overline{W}_i^m)$. The decoder reconstructs $f_m(e_{1,m}(\overline{W}_1^m),e_{2,m}(\overline{W}_2^m))$. The decoder recovers $(\overline{W}_1^m,\overline{W}_2^m)=(W_1^{n,m},W_2^{n,m})$ with probability of error vanishing in $m$. Encoder $i$ transmits $l_{1,i}$ bits, where $i\in \{1,2\}$.}
\\\textbf{Step 2.} \textcolor{black}{Constructing the distributed source with side-information $(\mathcal{X}_1,\mathcal{X}_2,{\mathcal{W}},{\mathcal{W}},P'_{X_1,X_2,W'_1,W_2},d_1,d_2)$:}
\\ At the end of Step 1, Encoder $i$ has access to $W_i^{n,m}, i\in\{1,2\}$ while the decoder has $(W_1^{n,m},W_2^{n,m})$. 
\textcolor{black}{It should be noted that the empirical distribution of $(S_1^{n,m},W_1^{n,m})$ is close, but not equal to $P_{S_1,W}$ as desired. In this step, we produce a matrix ${W'_1}^{n,m}$ at the decoder such that the empirical distribution of $(S_1^{n,m},{W'_1}^{n,m})$ is exactly equal to  $P_{S_1,W}$.}
Without loss of generality, let
\[ \mathcal{S}= \{1,2,\cdots,|\mathcal{S}|\}, \qquad \mathcal{W}= \{1,2,\cdots,|\mathcal{W}|\}.\]
\textcolor{black}{Let $T^{n}$ be a sequence of random variables, each  defined on the alphabet $\mathcal{T}\triangleq \{0,1,\cdots, |\mathcal{W}|\}$ in the following way. 
Each element $T(i)$ is produced conditioned on the element $S_1(i)$, and with conditional distribution $P'_{T|S_1}$, where the conditional distribution will be defined in the sequel. Note that $T^{n}$ is an i.i.d. sequence, since $S^{n}$ is an i.i.d. sequence. Denote the joint distribution of $(T(i),S_1(i))$ by $P'_{T(i),S_1(i)}=P'_{T,S_1}$.  Assume the following Markov chain $T^n- S_1^n- (X_1^n,X_2^n,W_1^n)$. The joint distribution of the triple 
($T (i),S_1(i), W_1(i)$) is given by $P'_{T(i),S_1(i),W_1(i)}=P_{S_1}P'_{T|S_1}P'_{W_1(i)|S_1(i)}$ for any  $i\in [1,n]$. Define the random variable $W_1'$ on the alphabet $\mathcal{W}$ as follows for all $i \in [1,n]:$}
\begin{align}
\textcolor{black}{ W'_1(i)=
 \begin{cases}
 W_1(i),\qquad&\text{ if } T(i)=0,\\
 T(i),&\text{ if } T(i)\neq 0.
\end{cases}}
\label{eq:s20}
\end{align}
Note that $W'_1$ is a function of the pair $(W_1,T)$. Let $F$ denote this function, i.e., $W_1'=F(W_1,T)$. \textcolor{black}{
In the sequel, we will define $P'_{T|S_1}$ such that $P'_{S_1(I),W'_1(I)}=P_{S_1,W}$,  where $I$ is uniformly distributed on $[1,n]$, and is chosen independently of $S_i,W_i,i\in\{1,2\}$.} First, note from Definition 
\ref{def:Q} that  we have:
\begin{align}
\textcolor{black}{P'_{S_1^n,W_1^n}\left(d^n_{S_1,W}(S_1^n,Q^*_{n,\phi,\phi'}(S_1^n))>\phi \right)\leq \phi'  \Rightarrow P'_{S_1^n,W_1^n}\left(\frac{1}{n}N(a,b|S_1^n,W_1^n)>P_{S_1,W}(a,b)+\phi \right)\leq \phi', \forall a\in \mathcal{S},b\in \mathcal{W}.}\label{eq:s21}
\end{align}
 Also, note that 
\begin{align}
\nonumber\mathbb{E}\left(N(a,b|S_1^n,W_1^n)\right)&=\mathbb{E}\left(\sum_{i=1}^n \mathbbm{1}(S_1(i)=a,W_1(i)=b)\right)=\sum_{i=1}^nP'_{S_1(i),W_1(i)}(a,b)\\
&\nonumber =n\sum_{i=1}^n\frac{1}{n}P'_{S_1(i),W_1(i)}\left(a,b\right)
\\&=nP'_{S_1(I),W_1(I)}(a,b),\label{eq:s22}
\end{align}
\textcolor{black}{where $P'_{S_1(I),W_1(I)}(a,b)$ can be interpreted as the average empirical distribution of $(S_1^n,W_1^n)$.}
Define $\mathcal{E}_{a,b}, \ (a,b)\in \mathcal{S}\times \mathcal{W}$ as the event that $\frac{1}{n}N(a,b|S_1^n,W_1^n)>P_{S_1,W}(a,b)+
\phi$. 
 Then, 
\begin{align}
\nonumber 
\textcolor{black}{\mathbb{E}\left(N(a,b|S_1^n,W_1^n)\right)}&
\textcolor{black}{=P'_{S_1^n,W_1^n}(\mathcal{E}_{a,b})\mathbb{E}(N(a,b|S_1^n,W_1^n)|\mathcal{E}_{a,b})+P'_{S_1^n,W_1^n}(\mathcal{E}_{a,b}^c)\mathbb{E}(N(a,b|S_1^n,W_1^n)|\mathcal{E}_{a,b}^c)}\\
 &\stackrel{(a)}{\leq} 
 \textcolor{black}{\phi'\cdot n+n( P_{S_1,W}(a,b)+\phi),}\label{eq:s23}
\end{align}
 where in (a) we have used Equation \eqref{eq:s21}. Similarly we can show the following:
\begin{align}
\label{eq:s24}
\textcolor{black}{\mathbb{E}(N(a,b|S_1^n,W_1^n))\geq (1-\phi')n( P_{S_1,W}(a,b)-\phi).}
\end{align}
\textcolor{black}{From Equations \eqref{eq:s22}, \eqref{eq:s23}, and \eqref{eq:s24}, we have for all $a,b\in \mathcal{S}\times \mathcal{W}$:}
\begin{align}
&\nonumber 
\textcolor{black}{(1-\phi')( P_{S_1,W}(a,b)-\phi)\leq P'_{S_1(I),W_1(I)}(a,b)\leq \phi+\phi' +P_{S_1,W}(a,b)}\\
\label{eq:s23.5}
&\textcolor{black}{\Rightarrow \forall (a,b)\in \mathcal{S}\times\mathcal{W}: \exists \gamma_{a,b}\in [-\phi-\phi'+\phi\phi',\phi+\phi']: P'_{S_1(I),W_1(I)}(a,b)=P_{S_1,W}(a,b)+\gamma_{a,b}.}
\end{align}
On the other hand from \eqref{eq:s20}:
\begin{align}
 P'_{S_1(I),W'_1(I)}(a,b) &= 
\sum_{(c,d):F(c,d)=b} P'_{S_1(I)}(a) P'_{W_1(I)|S_1(I)}(c|a) P'_{T(I)|S_1(I)}(d|a) \\
&=  P'_{S_1(I),T(I)}(a,b)+P_{S_1(I)}(a)P'_{T(I)|S_1(I)}(0|a)P'_{W_1(I)|S_1(I)}(b|a).
 \label{eq:s25}
\end{align}
So, in order for $P'_{S(I),W'(I)}(a,b)=P_{S_1,W}(a,b),\forall (a,b)\in \mathcal{S}\times\mathcal{W}$ to hold, we must have:
\begin{align}
 P_{S_1,W}(a,b)=P'_{S_1,T}(a,b)+P_{S_1}(a)P'_{T|S_1}(0|a)P'_{W_1(I)|S_1(I)}(b|a).
 \label{eq:s26}
\end{align}
From \eqref{eq:s23.5} and \eqref{eq:s26}, we have:
\begin{align*}
  P_{S_1,W}(a,b)= P'_{S_1,T}(a,b)+P'_{T|S_1}(0|a)(P_{S_1,W_1}(a,b)+\gamma_{a,b}).
\end{align*}
Define the following:
\begin{align}
 &P'_{T|S_1}(0|a)\triangleq\min_{b\in\mathcal{W}} \frac{P_{S_1,W}(a,b)}{P_{S_1,W}(a,b)+\gamma_{a,b}}
\label{eq:pp1}
 \\
 &P'_{S_1,T}(a,b)\triangleq P_{S_1,W}(a,b)-(P_{S_1,W}(a,b)+\gamma_{a,b})P'_{T|S_1}(0|a), \forall a,b\in \mathcal{S}\times\mathcal{W}.
 \label{eq:pp2}
\end{align}
Note that $P'_{S_1,T}$ defined above is a valid probability distribution because of the following arguments. Fix an arbitrary $a \in \mathcal{S}$.  Since $\gamma_{a,b} +P_{S_1,W}(a,b) \geq 0$, for all $(a,b) \in \mathcal{S} \times \mathcal{W}$, we have $P'_{T_1|S_1}(0|a)\geq0$. \textcolor{black}{
Suppose that $P'_{T|S_1}(0|a)>1$, then we have 
$P_{S_1,W}(a,b)>P_{S_1,W}(a,b)+\gamma_{a,b}$ 
for all $b \in \mathcal{W}$. Noting that $\sum_b \gamma_{a,b}=0$, we see that $P_{S_1}(a)>P_{S_1}(a)$ which is a contradiction.  Hence, we have 
$0 \leq P'_{T_1|S_1}(0|a) \leq 1$.
Since $\gamma_{a,b}+P_{S_1,W}(a,b)\geq0$ for all $b \in \mathcal{W}$, and from the definition of $P'_{T|S_1}(0|a)$, we have 
$P'_{S_1,T}(a,b) \geq 0$. 
Moreover, $\sum_{b \in \mathcal{W}} P'_{S_1,T}(a,b)=P_{S_1}(a)-P_{S_1}(a)P'_{T_1|S_1}(0|a)$, so that $\sum_{b \in \mathcal{W}\cup\{0\}}P'_{S_1,T}(a,b) =P_{S_1}(a)$. }

\textcolor{black}{Thus we have ensured that $P'_{S(I),W'_1(I)}=P_{S_1,W}.$}
\textcolor{black}{
Next, note that}
\begin{align*}
\textcolor{black}{P'_{T|S_1}(0|a)}&=\textcolor{black}{\min_{b\in\mathcal{W} }\frac{P_{S_1,W}(a,b)}{P_{S_1,W}(a,b)+\gamma_{a,b}}}\\
 &\geq \textcolor{black}{\frac{\min_{a,b}P_{S_1,W}(a,b)}{\min_{a,b}P_{S_1,W}(a,b)+\max_{a,b}\gamma_{a,b}}}\\
 &\stackrel{(a)}{\geq}   \textcolor{black}{\frac{\min_{a,b}P_{S_1,W}(a,b)}{\min_{a,b}P_{S_1,W}(a,b)+\phi+\phi'}}
\end{align*}
\textcolor{black}{
where in (a) we have used \eqref{eq:s23.5}.
Let $\mathcal{F}$ be the event that $\{T=0\}$. Then, from the grouping axiom of entropy, we have:}
\begin{align}
 \textcolor{black}{H(T)= h_b(P(\mathcal{F}))+H(T|\mathcal{F}^c)P(\mathcal{F}^c)\leq h_b \left(p(\tau)\right)+(1-p(\tau))\log{|\mathcal{W}|}=\Lambda_n(\tau),}
\label{Eq:Comment}
\end{align}
\textcolor{black}{where $p(\tau)
\triangleq
\frac{\min_{a,b} P_{S_1,W}(a,b)}{\min_{a,b} P_{S_1,W}(a,b)+\phi+\phi'}$, and
the last inequality follows from the relation $P(\mathcal{F})> \frac{1}{2}$ which is true if $\phi+\phi' < \min_{a,b} P_{S_1,W}(a,b)$. The conditions $n \in \mathcal{B}(\epsilon)$ and $\tau \in (\frac{1}{\sqrt{n}}\sigma,\sigma')$ are chosen such that $\phi < \frac{1}{2} \min_{a,b} P_{S_1,W}(a,b)$
and $\phi' < \frac{1}{2} \min_{a,b} P_{S_1,W}(a,b)$. 
This can be seen as follows: 
\begin{align*}
    &     |\mathcal{S}|\sqrt{\frac{1}{2n}\log_e{\frac{8|\mathcal{S}|}{\min_{a,b\in\mathcal{S}\times\mathcal{W}}P_{S,W}(a,b)}}}< \tau 
    \Rightarrow \phi' <  \frac{1}{2} \min_{a,b} P_{S_1,W}(a,b),\\
    &
    \tau < \frac{|\mathcal{S}||\mathcal{W}|}{2(|\mathcal{S}|+|\mathcal{W}|)}\min_{a,b\in\mathcal{S}\times\mathcal{W}}P_{S,W}(a,b)
    \Rightarrow 
    \phi <  \frac{1}{2} \min_{a,b} P_{S_1,W}(a,b).
\end{align*}
}

Encoder 1 produces an i.i.d. vector of random variables $T^{n,m}$. At this step, the encoder transmits $T^{n,m}$ losslessly. This can be done with rate $r'_1=\frac{1}{nm}H(T^{n,m})=H(T)$. So, there exists a sequence of encoding functions $e'_{1,n,m}:{\mathcal{T}}^{n,m}\to \{0,1\}^{l'_{2,1}}, n,m\in \mathbb{N}$, and decoding functions $f'_{n,m}:\{0,1\}^{l'_{2,1}}\to \mathcal{T}^{n,m}$ with rate $\frac{l'_{2,1}}{nm}=H(T)+\epsilon_{nm}$, where $\epsilon_{nm}\to 0$ as $nm\to \infty$, and $P'\left(f'_{n,m}\left(e'_{1,n,m}(T^{n,m})\right)\neq T^{n,m}\right)\to 0$ as $nm\to \infty$. Encoder $1$ transmits $e'_{1,n,m}(T^{n,m})$. The decoder reconstructs $f'_{n,m}(e'_{1,n,m}(T^{n,m}))$. At the end of this step, the decoder recovers $T^{n,m},W_1^{'n,m},W_2^{n,m}$ with probability of error vanishing in $nm$, where $W_1^{'n,m}$ is defined in \eqref{eq:s20}. \textcolor{black}{Encoder $1$ transmits $l'_{2,1}$ bits in this step.
So, $l'_{2,1}\leq mn\Lambda_n(\tau)$.
Define the following single-letter distribution: for all 
$(x_1,x_2,w'_1,w_2) \in \mathcal{X}_1 \times \mathcal{X}_2 \times \mathcal{W} \times 
\mathcal{W}$,
\begin{align}
    &P'_{X_1,X_2,W'_1,W_2}(x_1,x_2,w'_1,w_2) \triangleq \frac{1}{n} \sum_{i \in [1,n]} P'_{X_1(i),X_2(i),W'_1(i),W_2(i)}(x_1,x_2,w'_1,w_2),
       \label{eq:9}
    \end{align}
    where
    \begin{align}
P'_{X_1^n,X_2^n,W'^n_1,W^n_2}(x_1^n,x_2^n,{w'}^n_1,w_2^n)\triangleq  P^n_{X_1,X_2}(x^n_1,x^n_2) \sum_{t^n,w^n_1\in \mathcal{T}^n\times \mathcal{W}^n} 
\mathbbm{1}({w}^n_1=Q(s_1^n)) \mathbbm{1}(w^n_2=Q(s_2^n)) {P'}^n_{T|S_1}(t^n|s^n_1) \mathbbm{1}({w'}^n_1= F^n(w^n_1,t^n)),
\end{align}
for all $x_1^n,x_2^n,{w'}^n_1,w_2^n\in \mathcal{X}_1^n\times \mathcal{X}_2^n\times \mathcal{W}^n\times \mathcal{W}^n$, and $s_i^n=f_i(x_i^n), i\in \{1,2\}$, and $F^n(\cdot,\cdot)$ denotes the $n$-letter function obtained by applying the function $F(\cdot,\cdot)$ on the components of its arguments.}
Our objective is to construct i.i.d. sequences at the distributed transmitter and receiver terminals which follow this single-letter distribution. \textcolor{black}{Note that the output of the finite blocklength quantizer is a vector of variables which are not independent and identically distributed.} In order to make the elements of the sequence identically distributed, we perform a random permutation of the elements as follows. We produce i.i.d. outputs because in the next step, we plan to use $(W'_1,W_2)$ as side-information for DSC-SI, and the i.i.d. property facilitates the analysis.
\\\textbf{Step 3:} Define the following:
\begin{align}
& \textcolor{black}{ \widetilde{W}'_1([1,n],i)=W_1'(\pi_i([1,n]),i), \quad \widetilde{W}_2([1,n],i)=W_2(\pi_i([1,n]),i),} 
\\& \textcolor{black}{\widetilde{X}_{j}([1,n],i)=X_j(\pi_i([1,n]),i), \ \ i\in [1,m],j\in\{1,2\}.}
\end{align}
\begin{Claim}
 The $m$-length vector of random tuples $(\widetilde{X}_1(j,[1,m]),\widetilde{X}_2(j,[1,m]),\widetilde{W}'_1(j,[1,m]),\widetilde{W}_2(j,[1,m]))$ is a vector of i.i.d. random variables for any fixed $j\in[1,n]$. Furthermore, the probability distribution does not depend on $j$.
\textcolor{black}{ The distribution of $(\widetilde{X}_1(j,l),\widetilde{X}_2(j,l),\widetilde{W}'(j,l),\widetilde{W}_2(j,l))$ is equal to $P'_{X_1,X_2,W'_1,W_2}$ for all $j \in [1,n],l \in [1,m]$.} 
 \label{Claim:1}
\end{Claim}
\textit{Proof.}
\textcolor{black}{Please see Appendix \ref{App:Claim1}.}

 \textcolor{black}{Define }
\begin{align}
 \textcolor{black}{ P'_{X_1,X_2,W'_1,W_2,U_1,U_2}(x_1,x_2,w_1,w_2,u_1,u_2)= P'_{X_1,X_2,W'_1,W_2}(x_1,x_2,w_1,w_2)P_{U_1|X_1,W}(u_1|x_1,w_1)P_{U_2|X_2,W}(u_2|x_2,w_2).}
\label{eq:distpp}
\end{align}
 \textcolor{black}{Fixing the distributions $P_{W|S_1}$, $P_{U_1|X_1,W}$, and $P_{U_2|X_2,W}$  completely determines the induced joint probability distribution $P'_{X_1,X_2,W'_1,W_2,U_1,U_2}$ via the choice of 
 $Q^*_{n,\phi,\phi'}$,
 $P'_{T|S_1}$, the function $F(\cdot,\cdot)$, and the random permutations $\pi_i, i\in [1,m]$. 
 Hence, determining the RD region given in Theorem \ref{thm: MCDSC22} involves taking the union of RD tuples satisfying \eqref{eq:prembounds1} to \eqref{eq:prembounds4} over all $f_i$, $g_i$, $P_{W|S_1}$, $P_{U_1|X_1,W}$ and $P_{U_2|X_2,W}$, where $g_i, i \in \{1,2\}$ are the reconstruction functions used at the decoder.} 

\textcolor{black}{Optimizing the RD region described in Theorem \ref{thm: MCDSC22} over all $f_i$, $g_i$, $P_{W|S_1}$, $P_{U_1|X_1,W}$ and $P_{U_2|X_2,W}$, involves an optimization over multi-letter functions whose input alphabets are not bounded.} Consequently, this region is not computable. For this reason we provide a computable single-letter inner bound to this region which is given in Theorem \ref{thm:MC}.

\begin{Claim}
 $P'_{X_1,X_2,W'_1,W_2,U_1,U_2}$ satisfies the following conditions:
\\\textcolor{black}{1)  $P'_{W'_1,U_1|X_1,X_2}(w_1,u_1|x_1,x_2)=P_{W|S_1}(w_1|s_1)P_{U_1|X_1,W}(u_1|x_1,w_1), \forall x_1,x_2,w_1,u_1\in \mathcal{X}_1\times \mathcal{X}_2\times \mathcal{W}\times \mathcal{U}_1,$ where $s_1=f_1(x_1)$,
\\2)  $P'_{U_2|X_2,W_2}(u_2|x_2,w_2)=P_{U_2|X_2,W}(u_2|x_2,w_2), \forall u_2,x_2,w_2\in \mathcal{U}_2\times \mathcal{X}_2\times \mathcal{W}$,
\\3) $P'(W'_1=W_2,S_1=S_2)\geq p(\tau)(1-\delta_n)$,
\\4) $U_1 - (X_1,W'_1) - (X_2,W_2) - U_2$,
\\ 5)  $(W'_1,W_2) - (S_1,S_2) - (X_1,X_2).$}
\label{Claim:2}
\end{Claim}
\textit{Proof.}
\textcolor{black}{Please see Appendix \ref{App:Claim2}.}

\textcolor{black}{
Consider the $m$-length vector of random tuples $(\widetilde{X}_1(j,[1,m]),\widetilde{X}_2(j,[1,m]),\widetilde{W}'_1(j,[1,m]),\widetilde{W}_2(j,[1,m]))$, $j \in [1,n]$.
Define $P_{\widehat{X}_i|W'_1,W_2,U_1,U_2}(\hat{x}_i|w_1,w_2,u_1,u_2)=\mathbbm{1}(\hat{x}_i=g_i(w_1,w_2,u_1,u_2))$ for all $(w_1,w_2,u_1,u_2,\hat{x}_i) \in 
\mathcal{W} \times \mathcal{W} \times \mathcal{U}_1 \times \mathcal{U}_2 \times 
\mathcal{X}$, and $i=1,2$.
Consider the distributed source with side information $({\mathcal{X}}_1,{\mathcal{X}}_2,{\mathcal{W}},{\mathcal{W}},P'_{X_1,X_2,W'_1,W_2},d_1,d_2)$.  From Theorem \ref{thm:BTSI}, since $P'_{U_1,U_2|X_1,X_2,W'_1,W_2} P_{\widehat{X}_1,\widehat{X}_2|W'_1,W_2,U_1,U_2} \in  \mathcal{P}_{BTSI}$, the set of rate-distortion tuples $(r''_1,r''_2,D_1,D_2) \in [0,\infty)^4$ that  satisfy the following bounds are achievable:}
\begin{align}
 &r''_1\geq I(X_1;U_1|W'_1,W_2,U_2),\qquad r''_2\geq I(X_2;U_2|W'_1,W_2,U_1),
 \label{thm:41}\\
 &r''_1+r''_2\geq I(X_1,X_2;U_1,U_2|W'_1,W_2),\label{thm:42}\\
&D_i \geq \mathbb{E} (d_i(X_i,\widehat{X}_i)), \ 
\ i \in \{1,2\}.
\label{thm:43}
\end{align}
\textcolor{black}{So, there exists a sequence of encoding functions $e''_{1,m}:{\mathcal{X}}_1^m\times{\mathcal{W}}^{ m}\to \{0,1\}^{l''_{1,i}}$, and $e''_{2,m}:{\mathcal{X}}_2^m\times{\mathcal{W}}^{m}\to \{0,1\}^{l''_{2,i}}, m\in \mathbb{N}$, and reconstruction functions $f''_{i,m}:\{0,1\}^{l''_{1,i}}\times\{0,1\}^{l''_{2,i}}\times{\mathcal{W}}^m \times {\mathcal{W}}^{m}\to {\mathcal{X}}_i^m, i\in\{1,2\}$ such that $\frac{l''_{1,i}}{m}\leq r''_i+\epsilon_m$, 
$\frac{1}{m} \sum_{k=1}^m \mathbb{E}d_i\left(\widetilde{X}_i(j,k),\widehat{\widetilde{X}}_i(j,k) \right) \leq D_i+\epsilon_m$, $i\in \{1,2\}$, $j \in[1,n]$, where $\epsilon_m\to 0$ as $m\to \infty$, for any tuple $(r''_1,r''_2,D_1,D_2)$ satisfying \eqref{thm:41}, \eqref{thm:42}, and \eqref{thm:43}, where }
\[
\textcolor{black}{\widehat{\widetilde{X}}_i(j,[1:m])\triangleq f''_{i,m}\left(e''_{1,m}(\widetilde{X}_1(j,[1,m]),\widetilde{W}_1^{\prime} (j,[1,m])),e''_{2,m}(\widetilde{X}_2(j,[1,m]),\widetilde{W}_2(j,[1,m])), \widetilde{W}_1^{\prime}(j,[1,m]),\widetilde{W}_2(j,[1,m])\right).}
\]

\textcolor{black}{In this step, Encoder $1$ transmits $e''_{1,m}(\widetilde{X}_1(k,[1,m]), \widetilde{W}_1^{\prime}(k,[1,m])), k\in [1,n]$ in $n$ blocks each of length $m$}. Similarly, encoder $2$ transmits $e''_{2,m}(\widetilde{X}_2(k,[1,m]), \widetilde{W}_2(k,[1,m])), k\in [1,n]$. The decoder reconstructs the sources using the appropriate reconstruction functions. At the end of this step, the decoder recovers $\widehat{\widetilde{X}}_i^{n,m}, i\in {1,2}$. \textcolor{black}{Encoder $i$ transmits $nl''_{1,i}$ bits, where $i\in \{1,2\}$.}

\noindent  \textbf{Step 5.} Reconstructing the source at the decoder:
\\ The decoder reverses the permutation to recover the reconstruction of the source. \textcolor{black}{Define $\widehat{X}_i([1,n],j)=\widehat{\widetilde{X}}_i(\pi_j^{-1}([1,n]),j)$. Then from the previous steps, we conclude that for $i=1,2,$ }
\[
\textcolor{black}{\frac{1}{nm} \sum_{j=1}^n \sum_{k=1}^m \mathbb{E}(d_i(X_i(j,k), \widehat{X}_i(j,k)))\leq D_i+\epsilon_m.}
\]
Also, we have $R_1=r_1+r'_1+r''_1$, and $R_2=r_2+r''_2$. This gives the bounds in Theorem \ref{thm: MCDSC22}.

\end{proof}
\section{Proof of Theorem \ref{thm:MC}: Achievability of the FLMC Region}
\label{Sec:Ach}
In this section, we prove that the $\mathcal{RD}_{FLMC}$ region in Theorem \ref{thm:MC} is achievable. The proof uses the following lemmas. 

\begin{Lemma}[\cite{zhang2007estimating}]
 Let $P$ and $Q$ be probability measures defined on a common finite set $\mathcal{A}$, such that $V(P,Q)\leq v$, then:
 \begin{align*}
 \big|H(Q)-H(P)\big|\leq h_b(v)+v\log{(|\mathcal{A}|-1)}.
\end{align*}
\label{prop:RY}
\end{Lemma}
\begin{Lemma}
 Let $P_{XYZ}$ and $Q_{XYZ}$ be two probability measures defined on the finite alphabet $\mathcal{A}\times\mathcal{A}\times\mathcal{A}$, such that $V(P_{X,Y,Z},Q_{X,Y,Z})\leq v$, then:
 \\1) $V(P_{X,Y},Q_{X,Y})\leq v$,
\textcolor{black}{ \\2) $\big|I_P(X;Y)-I_Q(X;Y)\big|\leq 4\left(h_b(v)+v\log{|\mathcal{A}|}\right)$,
  \\3) 
  $\big|I_P(X;Y|Z)-I_Q(X;Y|Z)\big|\leq 8\left(h_b(v)+v\log{|\mathcal{A}|}\right)$.}
  \label{prop:mutI}
\end{Lemma}
\textit{Proof.}
\textcolor{black}{Please see Appendix \ref{App:prop:mutI}.}

The main idea in this proof is as follows. \textcolor{black}{
For every distribution $Q\in \mathcal{P}_{FLMC}$, 
we show that there exists a distribution $P'$ corresponding to $Q$, satisfying properties (1)-(5) in Theorem \ref{thm: MCDSC22}.
Using these properties,  we show that $V(P',Q)$ is no greater than 
$(1-p(\tau)+ p(\tau)\delta_n +\epsilon)$. Then, using the continuity of mutual information on finite alphabets as shown in Lemma \ref{prop:mutI}, Theorem \ref{thm:MC} follows from Theorem \ref{thm: MCDSC22}.} 

\textcolor{black}{Consider a distributed source $(P_{X_1,X_2},d_1,d_2)$, 
a pair of $\epsilon$-correlated components $(S_1,S_2)$ of the sources $X_1,X_2$,
 a distribution $Q_{W,U_1,U_2,\widehat{X}_1,\widehat{X}_2|X_1,X_2}\in \mathcal{P}_{FLMC}$, where $\widehat{X}_{i}=g_{Q,i}(W,U_1,U_2), i\in\{1,2\}$, a parameter $n \in \mathcal{B}(\epsilon)$, and  $\tau\in (\frac{1}{\sqrt{n}}\sigma, \sigma')$. Define the reconstruction functions  $g_{P',i}(W_1,W_2,U_1,U_2)\triangleq g_{Q,i}(W_1,U_1,U_2),$ $ i\in\{1,2\}$.
We claim that $\alpha_{FLMC}((S_1,S_2),Q,\tau,n)\subset\alpha_{MCML}(Q_{W|S_1},Q_{U_1|X_1,W},Q_{U_2|X_2,W},g_{P',1},g_{P',2},\tau,n)$. This along with Theorem \ref{thm: MCDSC22} proves Theorem \ref{thm:MC}. Recall from  Theorem \ref{thm: MCDSC22} that the distribution $P'_{W'_1,W_2,U_1,U_2,X_1,X_2}$
corresponding to the tuple $(Q_{W|S_1},Q_{U_1|X_1,W},Q_{U_2|X_2,W},$ $g_{P',1},g_{P',2})$  can be constructed such that 
\[
P'_{W'_1,W_2,U_1,U_2,X_1,X_2}=P_{X_1,X_2} P'_{W'_1,W_2,U_1,U_2|X_1,X_2}=
P_{X_1,X_2}P'_{W_1,W_2|X_1,X_2}Q_{U_1|X_1,W}Q_{U_2|X_2,W},
\]
and $P'(W'_1=W_2,S_1=S_2) \geq p(\tau) (1-\delta_n)$.
Define the distribution $Q_{W,W,X_1,X_2,U_1,U_2}(w_1,w_2,x_1,x_2,u_1,u_2)$  as $\mathbbm{1}(w_1=w_2)P_{X_1,X_2}(x_1,x_2)Q_{W,U_1,U_2|X_1,X_2}(w_1,u_1,u_2|x_1,x_2)$ for all $w_1,u_1,u_2,x_1,x_2\in \mathcal{W}\times \mathcal{U}_1\times \mathcal{U}_2 \times \mathcal{X}_1 \times \mathcal{X}_2$. }
%

\begin{Claim}
The variational distance between $P'$ and $Q$ is bounded as follows:
\[
\textcolor{black}{V(P'_{W'_1,W_2,X_1,X_2,U_1,U_2},Q_{W,W,X_1,X_2,U_1,U_2})\leq  1-p(\tau)+p(\tau)\delta_n+\epsilon.} \]
\label{claim:RY}
\end{Claim}

\textit{Proof.} \textcolor{black}{Please see Appendix \ref{App:Claim3}.}

\textcolor{black}{Assume that $(R_1,R_2,D_1,D_2)\in \alpha_{FLMC}((S_1,S_2),Q,\tau,n)$. Then,
\begin{align*}
R_1&\geq I_Q(X_1;U_1|U_2,W) + E_{n,\delta_n} +\Gamma_{n,\delta_n}(\tau)+\Lambda_n(\tau),\\
R_2&\geq I_Q(X_2;U_2|U_1,W) +E_{n,\delta_n}+\Gamma_{n,\delta_n}(\tau),\\
R_1+R_2&\geq I_Q(X_1,X_2;U_1,U_2,W)+E_{n,\delta_n}+\Gamma_{n,\delta_n}(\tau)+\theta_n(\tau)
+\Lambda_n(\tau).
\end{align*}
As a result, from Lemma \ref{prop:mutI}, and Claim \ref{claim:RY}, we have 
\begin{align*}
R_1 &\geq I_{P'}(X_1;U_1|U_2,W'_1,W_2) + E_{n,\delta_n}+\Lambda_n(\tau),\\
R_2 &\geq I_{P'}(X_2;U_2|U_1,W'_1,W_2) +E_{n,\delta_n},\\
R_1+R_2 &\geq I_{P'}(W'_1;S_1)+I_{P'}(X_1;U_1|W'_1,W_2)+ I_{P'}(X_2;U_2|W'_1,W_2)-I_{P'}(U_1;U_2|W'_1,W_2)+\theta_n(\tau)+E_{n,\delta_n}+\Lambda_n(\tau),
\end{align*}
where the last inequality follows from \[
\textcolor{black}{I_{P'}(X_1,X_2;U_1,U_2,W'_1,W_2)\geq I_{P'}(W'_1;S_1)+I_{P'}(X_1;U_1|W'_1,W_2)+ I_{P'}(X_2;U_2|W'_1,W_2)-I_{P'}(U_1;U_2|W'_1,W_2).}
\]} 
\textcolor{black}{For the distortion vector, we have for $i=1,2$:}
\textcolor{black}{\begin{align*}
 D_i\geq \mathbb{E}_{Q}\{d_i\big(X_i,g_{Q,i}(U_1,U_2,W)\big)\}+ 2\left(1-p(\tau)+p(\tau)\delta_n+\epsilon\right) d_{i,max}\Rightarrow D_i\geq \mathbb{E}_{P'}\{d_i\big(X_i,g_{P',i}(U_1,U_2,W'_1,W_2)\big)\},
\end{align*}}
\textcolor{black}{which is shown below:}
 \textcolor{black}{
\begin{align*}
&\mathbb{E}_{Q}\{d_i\big(X_i,g_{Q,i}(U_1,U_2,W)\big)\}= \sum_{u_1,u_2,w,x_i}Q_{W,X_i,U_1,U_2}(w,x_i,u_1,u_2)d_i\big(x_i,g_{Q,i}(u_1,u_2,w)\big)
\\&
=\sum_{u_1,u_2,w_1,w_2,x_i}
Q_{W,W,X_i,U_1,U_2}(w_1,w_2,x_i,u_1,u_2)d_i\big(x_i,g_{Q,i}(u_1,u_2,w_1)\big)
\\
 &\geq \sum_{u_1,u_2,w_1,w_2,x_i}P'_{W_1,W_2,X_i,U_1,U_2}(w_1,w_2,x_i,u_1,u_2)d_i\big(x_i,g_{P',i}(u_1,u_2,w_1,w_2)\big)
 \\&-d_{i,\max} \sum_{u_1,u_2,w_1,w_2,x_i}|Q_{W,W,X_i,U_1,U_2}(w_1,w_2,x_i,u_1,u_2)-P'_{W'_1,W_2,X_i,U_1,U_2}(w_1,w_2,x_i,u_1,u_2)|
 \\
 &\geq \mathbb{E}_{P'}\{d_i\big(X_i,g_{P',i}(U_1,U_2,W'_1,W_2)\big)\}-
 2\left(1-p(\tau)+p(\tau)\delta_n+\epsilon\right)d_{i,max}.
\end{align*}}
So, $(R_1,R_2,D_1,D_2)\in \alpha_{MCML}(Q_{W|S_1},Q_{U_1|X_1W},Q_{U_2|X_2W},g_{P',1},g_{P',2},\tau,n)$. This completes the proof.

\section{Binary One Help One Example}
\label{Sec:BOHO}
In this section, we investigate an example where the FLMC region is strictly larger than the CC region.
\begin{Example}
 \textcolor{black}{Consider the distributed source depicted in Figure \ref{fig:BOHO}.} \textcolor{black}{Let $p,\epsilon\in (0,\frac{1}{2})$, and let X, Z and E be independent Bernoulli random variables with $P(X=1)=0.5, P(Z=1)=p, P(E=1)=\epsilon$. The distributed source consists of two source variables  $X_1=X+E$ and $X_2=(X,Z)$, and the distortion functions are given by $d_1(x_1,\hat{x}_1)=0$, for all $x_1,\hat{x}_1\in \{0,1\}$,  and $d_2(x_2,\hat{x}_2)=d_H(\hat{x}+\hat{z},x+z)$ where $\hat{x}_2=(\hat{x},\hat{z}),x_2=(x,z)$, and $d_H$ is the binary Hamming distortion. This is called the binary-one-help-one (BOHO) problem.}
 \label{ex:BOHO}
\end{Example}
\textcolor{black}{The setup was investigated in \cite{wagner} for case when the two sources have non-trivial common components (i.e. $E=0$).} The first encoder is called the helper. The reason is that the distortion constraint is only a function of the second encoder's source and its reconstruction. \textcolor{black}{Hence, Encoder 1 is `helping' the second encoder by sending a quantized version of its source.} \textcolor{black}{The following gives an inner bound to the achievable RD region for this problem for the case when $\epsilon=0$.}
\begin{Lemma}[\cite{wagner}]
For $\epsilon=0$, the following RD quadruples are achievable using the CC scheme. 
\begin{equation}
\textcolor{black}{(r_1,r_2,d_1,d_2)=(1-h_b(\delta),h_b(p\ast\delta)-h_b(\delta_1),0,\delta_1), \qquad 0\leq \delta_1\leq \delta \leq 1.}
\label{bound}
\end{equation}
 \end{Lemma}
When $\epsilon\neq 0$, it was shown that \eqref{bound} is not achievable, i.e. $(r_1,r_2,d_1,d_2)\notin \mathcal{RD}^*(P_{X_1,X_2},d_1,d_2)$. 

\begin{figure}
\includegraphics*[draft=false,scale=.8]{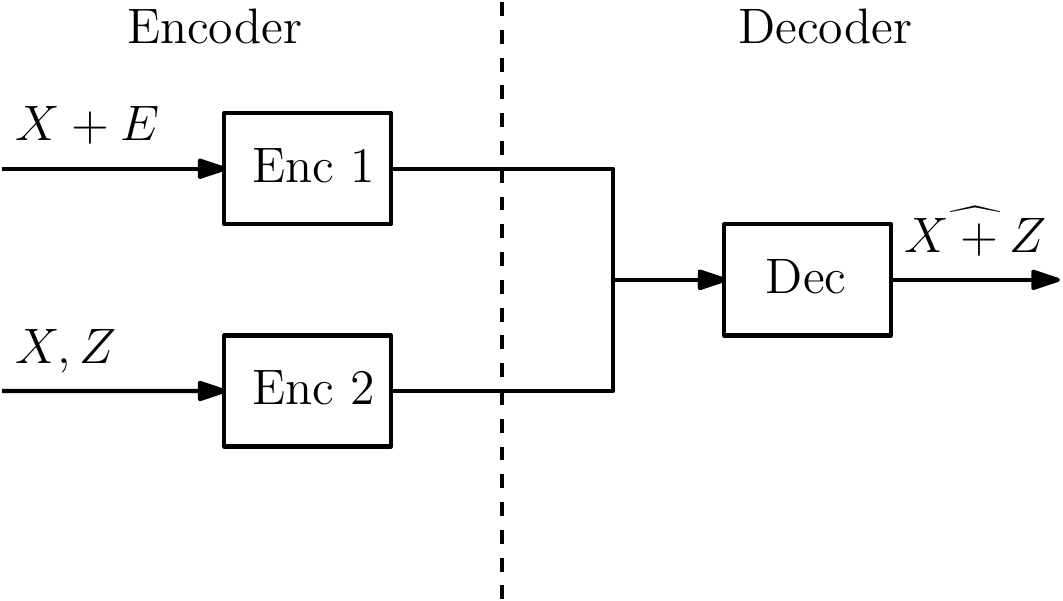}
\centering 
\caption{The binary one help one example}
\label{fig:BOHO}
\end{figure}

\subsection{Finite Length Quantizer Scheme} 
\textcolor{black}{In the previous sections, we introduced the FLMC coding scheme and proved an inner bound to the achievable RD region for this scheme. Note that when $\epsilon\neq 0$, then if $(S_1,S_2)$ is chosen to be trivial (i.e. constant variables), then the FLMC region is the same as the CC region. 
In this section, we investigate the FLMC scheme  when the choice of correlated components is  $(S_1,S_2)= (X+E,X)$, and the distribution $P_{S_1,W}$ is chosen such that $W$ is a quantized version of $S_1$ using a binary symmetric test channel. We show that the scheme achieves RD tuples outside of the $\mathcal{RD}_{CC}$ region.}


\begin{Theorem} 
\textcolor{black}{
For any $\delta \in (0,\frac{1}{2})$,
$n \in \left(\frac{64}{\delta^2} \ln \frac{32}{\delta}, \frac{\log(2 \epsilon)}{\log(1-\epsilon)} \right)$, $\tau \in \left(2 \sqrt{\frac{1}{n} \ln \frac{32}{\delta}},\frac{\delta}{4} \right)$,
$\delta_1\in (0,p\ast \delta')$, let $\alpha_{BOHO\mbox{-}FLMC}(\delta,\delta_1,n,\tau)$ denote the set of RD vectors $(R_1,R_2,0,D_2)$ satisfying the following:
\begin{align}
&R_1\geq 1-h_b(\delta)+\theta'_n,\\
& R_2\geq h_b(p\ast \delta')-h_b(\delta_1),\\
&D_2\geq \delta_1+\delta_n \left(\delta'+\frac{\epsilon}{\delta_n}\ast\delta'
\right),
\label{rate}
\end{align}
 where $a\ast b=a(1-b)+b(1-a)$, $\delta_n=(1-(1-\epsilon)^n)$ and 
 $\theta'_n=\frac{1}{n} \log \left( \frac{n \tau^2}{4}-\ln 4 \right)+\tau \left( 4+h(\delta)-\log (\delta(1-\delta)) \right)+\frac{3}{n}$,
 and $\delta'=\min\{1,\delta+\tau+8 e^{-2n \tau^2/4}\}$.}
\textcolor{black}{The following RD region is achievable:
\begin{align*}
    \mathcal{RD}_{BOHO\mbox{-}FLMC}\triangleq 
    \bigcup_{\delta\in (0,\frac{1}{2})}
    \bigcup_{n\in \left(\frac{64}{\delta^2} \ln \frac{32}{\delta}, \frac{\log(2 \epsilon)}{\log(1-\epsilon)} \right)}
\bigcup_{\tau \in \left(2 \sqrt{\frac{1}{n} \ln \frac{32}{\delta}},\frac{\delta}{4} \right) }
   \bigcup_{\delta_1 \in (0,p\ast \delta')}
    \alpha_{BOHO\mbox{-}FLMC}(\delta,\delta_1,n,\tau).
\end{align*}}
\label{thm:RDDSC}
\end{Theorem}

\begin{proof}
The proof follows steps which are similar to the proof of Theorem \ref{thm:MC}. However, it is useful to investigate the details of these steps for this particular example to clarify the scheme.

\textbf{Codebook Generation:}  Fix $n,m\in \mathbb{N}$. The coding scheme uses two codebooks $\mathcal{C}_f^n$ and $\mathcal{C}_r^m$. \textcolor{black}{The first codebook  $\mathcal{C}_f^n$ corresponds to the quantizer $Q^*_{n,\phi, \phi'}$, where $P_{S_1,W}$ corresponds to the source and quantized output in the point-to-point quantization of a binary symmetric source with Hamming distortion $\delta$ (i.e. $W=S_1\oplus_2 N_{\delta}$, where $N_{\delta}$ is Bernoulli with parameter $\delta$ and is independent of $S_1$). Let $W^n\triangleq Q^*_{n,\phi, \phi'}(S_1^n)$. The average Hamming distortion achieved is given by:}
\begin{align*}
 \textcolor{black}{   \frac{1}{n} \mathbb{E}(d_H(S_1^n, W^n))     =
    \frac{1}{n}\left(P(\mathcal{F}) \mathbb{E}(d_H(S_1^n, W^n)|\mathcal{F})+
        P(\mathcal{F}^c) \mathbb{E}(d_H(S_1^n, W^n)|\mathcal{F}^c)\right)
        \leq \min\{1,\delta+\tau+ 8 e^{-2n \tau^2/4}\}=\delta',}
\end{align*}
\textcolor{black}{where $\mathcal{F}$ is the event that $l_{\infty}(P_{S_1,W}, \mathbb{P}_{S_1^n, W^n})>\tau$.
 The second codebook $\mathcal{C}_r^{(m)}$, corresponds to a standard randomly generated $m$-length quantizer operating on a Bernoulli source with parameter $p\ast \delta'$, and expected Hamming distortion of $\delta_1$, and $m$ is taken to be asymptotically large. The rate of the codebook is $h_b(p\ast \delta')-h_b(\delta_1)+\lambda_m$, where $\lambda_m\to 0$ as $m\to \infty$.}
\begin{figure}
\includegraphics*[draft=false,width= \columnwidth]{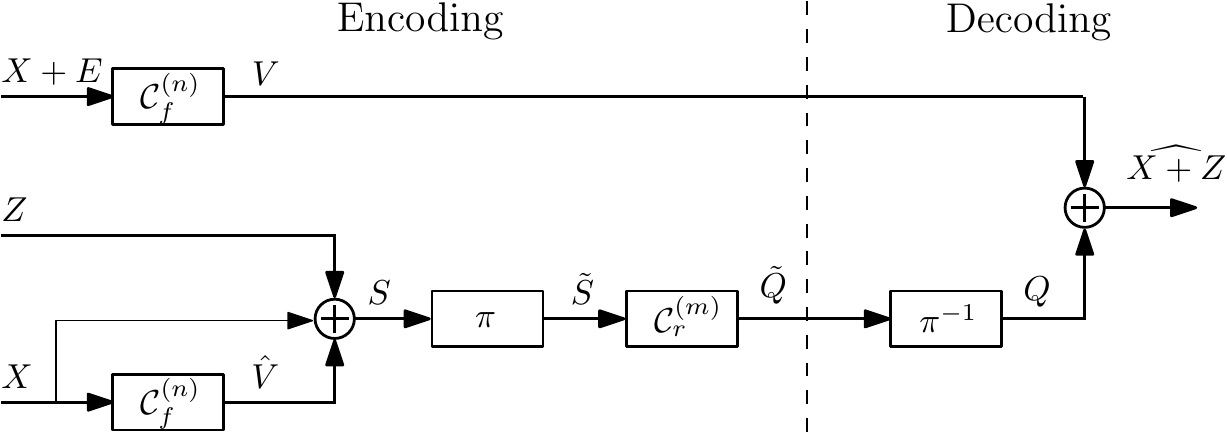} 
\caption{A block-diagram of the scheme}
\label{block}
\end{figure}
Let $\mathcal{S}_n$ be the set of permutations on the set $[1,n]$. The permutations $\pi_i ,i\in [1,m]$ are chosen randomly, independently and uniformly from $\mathcal{S}_n$. 

\textcolor{black}{First, we provide an outline of the encoding and decoding steps. Figure \ref{block} provides a block diagram of the coding scheme.} Assume that the encoders receive an $n\times m$ block of the source vectors $X_i([1,n],[1,m])$, $i \in{1,2}$. \textcolor{black}{ The first encoder uses $\mathcal{C}_f^{(n)}$ to quantize $m$ blocks of length $n$ of its source sequence $(X+E)([1,n],j),j \in [1,m]$ into the vectors $V([1,n],j)$. It sends the vectors $V([1,n],j)), j\in [1,m]$ to the decoder.} The second encoder `guesses' the quantization sent by the first encoder. More precisely, it uses $C_{f}^n$ to quantize the blocks of $n$-length vectors of the source $X([1,n],j), j\in [1,m]$ into $\widehat{V}([1,n],j)$. \textcolor{black}{Note that the two vectors ${V}([1,n],j)$ and $\widehat{V}([1,n],j)$ are equal with high probability if  $\epsilon\ll\frac{1}{n}$. The reason is that $P\left({V}([1,n],j)\neq \widehat{V}([1,n],j)\right)=(1-(1-\epsilon)^n)\approx n\epsilon$.} Next, the second encoder finds the quantization noise by calculating the binary addition of $X([1,n],j)$ and $\widehat{V}([1,n],j)$ for $j\in[1,m]$. \textcolor{black}{This represents the second encoder's estimate of the quantization noise of the quantizer used by the first encoder. In the next step, the second encoder calculates the binary addition $S([1,n],[1,m])\triangleq(X+\widehat{V}+Z)([1,n],[1,m])$. Then, the encoder permutes the result to get $\widetilde{S}([1,n],j)\triangleq\pi_j(S([1,n],j),j\in [1,m]$. We show in the sequel that $\widetilde{S}(i,[1,m]), i\in[1,n]$, is an i.i.d. vector of Bernoulli variables with parameter $p\ast\delta'$.} This i.i.d. source is quantized using $C_{r}^m$, and sent to the decoder. Let the quantized version of $\widetilde{S}$ be denoted by $\widetilde{U}$. The decoder applies the inverse of the permutations to get $U$, which is a lossy compression of $X+E+Z$. The decoder declares $U+V$ as the lossy reconstruction of $X+Z$.   

We proceed to formally present the scheme. The first encoder receives a string of $n\times m$ bits $X_1([1,n],[1,m])=(X+E)([1,n],[1,m])$. The encoder breaks this vector into $m$ blocks of length $n$. We denote each bit in this string by $X(i,j)+E(i,j), i\in[1,n],j\in[1,m]$, where $j$ indicates the block containing the bit, and $i$ indicates the index of the bit in the block. 
 The encoder uses $C_f^{(n)}$ to quantize each of these blocks. More precisely, it finds $V([1,n],j)$ such that:
\begin{equation*}
V([1,n],j)=\argmin_{v^n\in\mathcal{C}_f^{(n)}}d_H\left(X_1([1,n],j),v([1,n])\right),
\end{equation*}
for each $j\in [1,m]$.
\textcolor{black}{ The index of $V([1,n],j)$ in $C_f^{(n)}$ is sent to the decoder.} 
 \textcolor{black}{The rate of transmission for this encoder is $R_1=R_f^{(n)}=1-h_b(\delta)+\theta'_n$.}

 \textcolor{black}{The second encoder receives $n\times m$ pairs of bits $(X,Z)([1,n],[1,m])$. It divides each string into $m$ blocks of length $n$. It quantizes each block of $X([1,n],j), j\in [1,m]$, using $C_f^n$ similar to the first encoder.} Let $\widehat{V}([1,n],j)$ be the quantized codeword corresponding to $X([1,n],j)$. The encoder computes $S([1,n],[1,m])=X([1,n],[1,m])+\widehat{V}([1,n],[1,m])+Z([1,n],[1,m])$. Let $\widetilde{S}(i,j)=S(\pi_j(i),j),i\in[1,n],j\in[1,m]$.\textcolor{black}{ The next lemma proves that $\widetilde{S}(i,[1,m])$ is an i.i.d. vector of random variables with parameter $p\ast\delta'$
for any $i\in [1,n]$. 
\begin{Lemma} $\widetilde{S}(i,[1,m])$ is a string of i.i.d. Bernoulli random variables with parameter $p\ast\delta'$.
\label{Lem: Fin1}
\end{Lemma}
\textcolor{black}{\textit{Proof. Please see Appendix \ref{App:Lem: Fin1}}.}}

\textcolor{black}{ The encoder quantizes each $\widetilde{S}(i,[1,m]), i\in [1,n]$, using the code $C_r^{(m)}$. Let $\widetilde{U}(i,[1,m])$ be the quantization of $\widetilde{S}(i,[1,m])$. The encoder transmits the index of $\widetilde{U}(i,[1,m]), i\in [1,n]$, in $C_r^{(m)}$ to the decoder.} Define $\widetilde{T}([1,n],[1,m])=\widetilde{U}([1,n],[1,m])+\widetilde{S}([1,n],[1,m])$ as the quantization noise resulting from quantizing $\widetilde{S}([1,n],[1,m])$ to $\widetilde{U}([1,n],[1,m])$. 

 The decoder computes $U(i,j)=\widetilde{U}(\pi^{-1}_j(i),j), i\in[1,n],j\in[1,m]$. That is, the decoder undoes the permutation. Note that $\mathbb{E}(d_H(U(i,j),S(i,j)))=\mathbb{E}(d_H(\widetilde{U}(i,j),\widetilde{S}(i,j)))=\mathbb{E}(w_H(T(i,j)))=\delta_1$, by construction. The decoder declares $U([1,n],[1,m])+V([1,n],[1,m])$ as the reconstruction of the source sequence $(X+Z)([1,n],[1,m])$. The resulting average distortion is:
\begin{equation*}
D_2=\frac{1}{mn}\mathbb{E}\big\{d_H\left(\left(X+Z\right)\left([1,n],[1,m]\right),\left(U+V\right)([1,n],[1,m])\right)\big\},
\end{equation*}
\textcolor{black}{where $d_H(A([1,n],[1,m]), B([1,n],[1,m]))\triangleq \sum_{i=1}^n\sum_{j=1}^m d_H(A(i,j), B(i,j))$.}
This can be computed as follows:
\begin{align*}
&\mathbb{E}\{d_H((X+Z)([1,n],[1,m]),(U+V)([1,n],[1,m]))\}
=\mathbb{E}\{w_H((X+Z+S+T+V)([1,n],[1,m]))\}\\
&=\mathbb{E}\{w_H((\widehat{V}+V+T)([1,n],[1,m]))\}
\stackrel{(a)}{\leq} mn\left(\delta_1+ \frac{1}{mn}\mathbb{E}\{w_H((\widehat{V}+V)([1,n],[1,m]))\}\right),
\end{align*}
where (a) holds since for binary variables $A$ and  $B$, we have $w_H(A+B)\leq w_H(A)+w_H(B)$. We have:
\begin{align*}
\sum_{i=1}^m \mathbb{E}\{w_H((\widehat{V}+V)([1,n],i))\}
&\stackrel{(b)}{=} m(\mathbb{E}\{w_H((\widehat{V}\!+\!V\!)([1,n])|E([1,n])\!=\!0)P(E([1,n])\!=\!0)\}\\
&\hspace{1in} +\mathbb{E}\{w_H((\widehat{V}\!+\!V)([1,n])|E([1,n])\neq0)P(E([1,n])\!\neq\!0))\})\\
&\stackrel{(c)}{=} m(\mathbb{E}\{w_H((\widehat{V}\!+\!V)([1,n])|E([1,n])\!\neq\!0)P(E([1,n])\neq0))\\
&=m(1\!-\!(1\!-\!\epsilon)^n)\mathbb{E}\{w_H((\widehat{V}\!+\!V\!)([1,n]))|E([1,n])\neq0\}\\
&=m(1\!-\!(1\!-\!\epsilon)^n)\mathbb{E}\{w_H((X\!+\!\widehat{V}\!+\!X\!+\!V)([1,n]))|E([1,n])\!\neq\!0\}\\
&\leq m(1-(1-\epsilon)^n)\mathbb{E}\{w_H((X+\widehat{V})^n)|E([1,n])\neq0\}\\
&\hspace{1in} +\mathbb{E}\{w_H\left((X+E+V+E)([1,n])\right)|E([1,n])\neq0\}\\
&\stackrel{(d)}{=} m(1-(1-\epsilon)^n)\left(\mathbb{E}\{w_H((X+\widehat{V})^n)\}\right)\\
&\hspace{1in}+\frac{\epsilon}{(1-(1-\epsilon)^n)}\ast \mathbb{E}\{w_H((X+E+V)([1,n])\}\\
&\stackrel{(e)}{=}\textcolor{black}{mn\delta_n\left(\delta'+\frac{\epsilon}{\delta_n}\ast\delta'\right)}\\
&\textcolor{black}{\Rightarrow D_2\leq \delta_1+\delta_n\left(\delta'+\frac{\epsilon}{\delta_n}\ast\delta'\right).}
\end{align*}
 \textcolor{black}{ (b) holds since each block is quantized identically, and hence the expected value is equal for all blocks; (c) holds since if $E^n=0$ then $V^n=\widehat{V}^n$; (d) holds since $(X+E+V)^n$ is a function of $(X+E)^n$ which is independent of $E^n$, since $X^n$ is Bernoulli with parameter $\frac{1}{2}$ and $X$ and $E$ are independent, and finally (e) holds since $C_f^n$ corresponds to a $Q^*_{n}$ quantizer with Hamming distortion $\delta'$. This completes the proof of Theorem \ref{thm:RDDSC}.}
\end{proof}


\begin{figure}
\includegraphics*[draft=false,scale=.4]{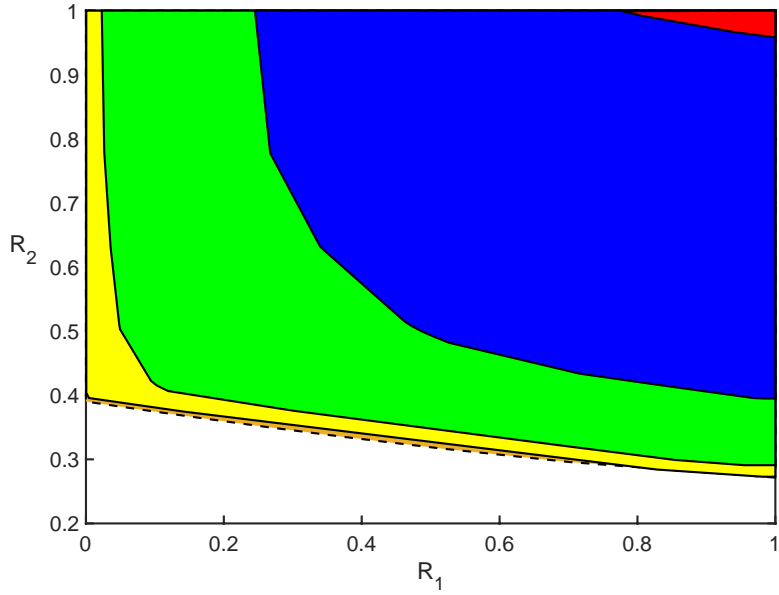}
\centering 
\captionsetup{font={stretch=1}}
\caption{
\textcolor{black}{
The achievable region using the scheme described in Theorem \ref{thm:RDDSC} for source parameters $p=0.3$ and $\epsilon= 10^{-4}, 10^{-5}, 10^{-7}, 10^{-14},0  $ in the red, blue, green, yellow, and orange (with dotted borderline) regions, respectively, and distortion less than or equal to $D_2=0.15$. Note that the achievable region becomes larger and converges to that of the CC scheme (orange region) as $\epsilon \to 0$.}} 
\label{fig:BOHO_RD}
\end{figure}
\textcolor{black}{
The achievable RD region of Theorem \ref{thm:RDDSC} is shown in Figure \ref{fig:BOHO_RD} for source parameters  $p=0.3$ and $\epsilon= 10^{-4}, 10^{-5}, 10^{-7}, 10^{-14},0$, and distortion $D_2=0.15$. It can be observed that the achievable region becomes larger and converges to that of the CC scheme (orange region) as $\epsilon \to 0$. The $\mathcal{RD}_{CC}$ region is not continuous in $\epsilon$ as shown in \cite{wagner}. Hence, the region is strictly contained in $\mathcal{RD}_{BOHO\mbox{-}FLMC}$. 
The next proposition formalizes this statement. }

\begin{Proposition}\textcolor{black}{ The RD region $\mathcal{RD}_{CC}$ is a strict subset of $\mathcal{RD}_{BOHO\mbox{-}FLMC}$, i.e. $\mathcal{RD}_{CC}\subsetneq \mathcal{RD}_{BOHO\mbox{-}FLMC}$.}
\end{Proposition}
\begin{proof} We show the existence of $p$ and $\epsilon$ such that the rate-distortion region in Theorem \ref{thm:RDDSC} strictly contains the CC region. It was shown in \cite{wagner} that when $\epsilon=0$, the BT region does not include the set of quadruples $(r_1,r_2,d_1,d_2)=(1-h_b(\delta),h_b(p\ast\delta)-h_b(\delta_1),0,\delta_1)$ for $\delta\in(0,0.5)$ and $\delta_1<p\ast\delta$. Also, the CC region reduces to the BT region when $\epsilon\neq 0$, since there is no common component in the sources $X_1$ and $X_2$.  Since the BT region can only be smaller when $\epsilon\neq0$ as compared to the case when $\epsilon=0$, we conclude that it does not include $(1-h_b(\delta),h_b(p\ast\delta)-h_b(\delta_1),0,\delta_1)$ for $\epsilon\neq 0$. So, for a given $\delta$ and $\delta_1$ there exists a radius $\gamma>0$ for which no quadruple in the set $\mathcal{B}((r_1,r_2,d_1,d_2),\gamma)=\{(R_1,R_2,0,D_2)\!:\!d_{E}((R_1,R_2,D_2),(r_1,r_2,d_2))\leq\gamma\}$ is in $\mathcal{RD}_{CC}$,
where $d_{E}$ is the Euclidean distance in the three dimensional space.
\textcolor{black}{However, from Theorem \ref{thm:RDDSC},  $(r\rq{}_1,r\rq{}_2,0,d\rq{}_2)=(1-h_b(\delta)+\theta'_n,h_b(p\ast\delta')-h_b(\delta_1),0,
\delta_1+\delta_n(\delta'+\frac{\epsilon}{\delta_n}\ast\delta'))$ is achievable. We have:
\begin{align}\label{teta}
d_{E}((r\rq{}_1,r\rq{}_2,d\rq{}_2),(r_1,r_2,d_2))=\sqrt{{\theta'_n}^2+ \left(h_b(p\ast \delta')- h_b(p\ast \delta)\right)^2+
\delta^2_n\left(\delta'+
\frac{\epsilon}{\delta_n}\ast \delta' \right)^2}.
\end{align}
Note that $\theta'_n\to 0$ and $\delta'\to \delta$ as $n\to\infty$ for $\tau= n^{-\frac{1}{2}+\alpha}, 0<\alpha<\frac{1}{2}$.  
Hence there exists $n$ large enough that 
$\theta'^2_n+(h_b(p \ast \delta')-h_b(p \ast \delta))^2 \leq (\frac{\gamma}{2})^2$. 
Next note that for this value of $n$, the term 
$\delta_n^2\left(\delta'+ (\epsilon/\delta_n) \ast \delta')\right)^2 \searrow 0$
as $\epsilon \searrow 0$. So there exists $\epsilon>0$ sufficiently small for which $n<\frac{\log{2\epsilon}}{\log{1-\epsilon}}$ and  $(r\rq{}_1,r\rq{}_2,0,d\rq{}_2)\in \mathcal{B}((r_1,r_2,0,d_2),\gamma)$. This completes the proof.}
\end{proof}

\section{Conclusions}
\label{Sec:Conc}
We have considered the problem of lossy distributed data compression. We have proposed a new coding scheme which uses a two layered coding approach. The first coding layer consists of randomly generated codebooks with constant finite blocklengths. The second coding layer consists of random unstructured codes with asymptotically large blocklengths. A computable expression for an inner bound to the achievable RD region is derived which is expressed in terms of information measures which  are functions of single-letter probability distributions. We have shown that the resulting RD region strictly contains the best-known achievable regions in the literature. 
%
%

\noindent \textbf{Acknowledgements:} We thank the anonymous reviewers for their careful reading and critical comments that helped us in creating a better manuscript. 

\begin{appendices}
\section{Proof of Proposition \ref{Prop:Subset}}
\label{App:Subset}
\textcolor{black}{Let $(R^*_1,R^*_2,\Delta^*_1,\Delta^*_2)\in \mathcal{RD}_{CC}(P_{X_1,X_2},d_1,d_2)$. It is enough to show that for any $\gamma>0$, we have $(R_1^*+\gamma,R_2^*+\gamma,\Delta^*_1+\gamma,\Delta^*_2+\gamma)\in \mathcal{RD}_{FLMC}(P_{X_1,X_2},d_1,d_2)$. To show this, let $S$ be a common component of the sources $(X_1, X_2)$, and let $P_{W,U_1,U_2,\widehat{X}_1,\widehat{X}_2|X_1X_2}\in \mathcal{P}_{CC}$ as in Theorem \ref{thm:CC}, achieving $(R^*_1,R^*_2,\Delta^*_1,\Delta^*_2)$. We find a pair of $\epsilon$-correlated components $(S'_1,S'_2)$ for some $\epsilon\in [0,1]$, and construct a distribution $P'_{W,U_1,U_2,\widehat{X}_1,\widehat{X}_2|X_1,X_2}\in \mathcal{P}_{FLMC}$ as in Theorem \ref{thm:MC}, achieving $(R^*_1+\gamma,R^*_2+\gamma,\Delta^*_1+\gamma,\Delta^*_2+\gamma)$. Let $S,U_{1},U_{2},W, \widehat{X}_1$, and $\widehat{X}_2$ be the random variables corresponding to the distribution $P_{W,U_1,U_2,\widehat{X}_1,\widehat{X}_2|X_1,X_2}$.  We use these random variables to find suitable random variables $S'_1,S'_2, U'_1, U'_2, W', \widehat{X}'_1,\widehat{X}'_2$ corresponding to $P'_{W,U_1,U_2,\widehat{X}_1,\widehat{X}_2|X_1,X_2}$.  Let $S'_1=S'_2=S$ and $U'_1=U_{1}, U'_2=U_{2}$,$W'=W$, and $\widehat{X}_i=\widehat{X}'_i$,  for $i\in \{1,2\}$.
Then, $(S'_1,S'_2)$ is a pair of $0$-correlated components of $(X_1,X_2)$. Also, the Markov chain conditions  $U'_1 - (X_1,W') - (X_2,W') - U'_2$  and  $W'-S'_1-(X'_1,X'_2)\ $ are satisfied, since  $U_1 - (X_1,W) - (X_2,W) - U_2$  and  $W-S_1-(X_1,X_2)$. As a result, we have $P'_{W,U_1,U_2,\widehat{X}_1,\widehat{X}_2|X_1,X_2}\in \mathcal{P}_{FLMC}$. Similarly, it can be seen that the corresponding mutual information terms $I(X_1;U_1|U_2W)$, $I(X_2;U_2|U_1W)$, $I(X_1,X_2;U_1,U_2|W)$ and the expected distortions $\mathbb{E}\{d_i\big(g_i(U_1,U_2,W),X_i\big)\}, i\in \{1,2\}$
are the same for both distributions $P$ and $P'$. 
  Note that, $P(S'_1\neq S'_2)=\epsilon=0$. Hence, $\delta_n=1-(1-\epsilon)^n=0$. So, we get $E_{n,\delta_n}=0$. 
Note that all of the other terms on the right hand side of equations \eqref{eq:01}-\eqref{eq:04} which depend on $n$ approach $0$ as $n\to \infty$.
Consequently, for any $\gamma>0$, there exists an $n'$ such that $(R^*_1+\gamma,R^*_2+\gamma,\Delta^*_1+\gamma,\Delta^*_2+\gamma)$ satisfies the bounds in equations \eqref{eq:01}-\eqref{eq:04}. This completes the proof.}
\section{Proof of Proposition \ref{prop:fin}}
\label{App:prop:fin}
\textcolor{black}{The proof uses the following definitions of typical, joint typical, and conditional typical sets:}
\textcolor{black}{
\begin{Definition}[\textbf{Typical Sequence}]
Given a probability space $(\mathcal{X},P_X)$, a blocklength $n\in \mathbb{N}$, and a parameter $\zeta>0$, the $\zeta$-typical set $\mathcal{A}^n_{\zeta}(X)$ is defined as: 
\begin{align*}
    \mathcal{A}^n_{\zeta}(X)
    = \Big\{x^n: \left| \frac{1}{n} N(a|x^n) -P_X(a) \right| \leq \frac{\zeta}{|\mathcal{X}|} \mathbbm{1}(P_{X}(a) \neq 0), \quad  
\forall a \in \mathcal{X}\Big\}.
\end{align*}
\end{Definition}
\begin{Definition}[\textbf{Joint Typical Sequence}]
Given a probability space $(\mathcal{X}\times \mathcal{Y},P_{X,Y})$, a blocklength $n\in \mathbb{N}$, and a parameter $\zeta>0$, the joint $\zeta$-typical set $\mathcal{A}^n_{\zeta}(X,Y)$ is defined as: 
\begin{align*}
    \mathcal{A}^n_{\zeta}(X,Y)
    = \Big\{(x^n,y^n): \left| \frac{1}{n} N(a,b|x^n,y^n) -P_{X,Y}(a,b) \right| \leq \frac{\zeta}{|\mathcal{X}||\mathcal{Y}|}\mathbbm{1}(P_{X,Y}(a,b) \neq 0), \quad  
\forall a,b \in \mathcal{X}\times \mathcal{Y}\Big\}.
\end{align*}
\end{Definition}
\begin{Definition}[\textbf{Conditional Typical Sequence}]
Given a conditional distribution $P_{Y|X}: \mathcal{X} \rightarrow
\mathcal{Y}$, a blocklength $n\in \mathbb{N}$, a parameter $\zeta>0$, and a sequence $x^n \in \mathcal{X}^n$, the conditional $\zeta$-typical set of $\mathcal{A}^n_{\zeta}(Y|x^n)$  is defined as:
\begin{align*}
 \mathcal{A}^n_{\zeta}(Y|x^n)= \Big\{y^n:\left| \frac{1}{n} N(a,b|x^n,y^n)-\frac{1}{n} N(a|x^n) P_{Y|X}(b|a) \right|
\leq \frac{\zeta}{|\mathcal{Y}|}\mathbbm{1}(P_{Y|X}(b|a) \neq 0), \quad \forall a,b\in \mathcal{X}\times \mathcal{Y}\Big\}.
\end{align*}
\end{Definition}
The following lemma describes a relation between the typical, joint typical and conditional typical sets which will be used in the next steps:
\begin{Lemma}[\cite{csiszarbook}]
Given a probability space $(\mathcal{X}\times \mathcal{Y},P_{X,Y})$,  blocklength $n\in \mathbb{N}$, and parameter $\zeta>0$, let $x^n\in \mathcal{A}_{\zeta}^n(X)$ and $y^n \in  \mathcal{A}^n_{\zeta}(Y|x^n)$. Then, $(x^n,y^n) \in \mathcal{A}^n_{\hat{\zeta}}(X,Y)$, where $\hat{\zeta}= \zeta(|\mathcal{X}|+|\mathcal{Y}|)$. 
\label{Lem:5}
\end{Lemma}
The proof follows from the triangle inequality, and is omitted for brevity. It should be noted that a consequence of Lemma 5 is that if $x^n\in \mathcal{A}_{\zeta}^n(X)$ and $y^n \in  \mathcal{A}^n_{\zeta}(Y|x^n)$, then $y^n \in \mathcal{A}^n_{\hat{\zeta}}(Y)$.
Furthermore, we will use the following results on the cardinality and probability of the typical and conditional typical sets.} 
\textcolor{black}{
\begin{Lemma}[\cite{csiszarbook}]
\label{Lem:6}
Given a probability space $(\mathcal{X},P_X)$, a blocklength $n \in \mathbb{N}$, and  a parameter $\zeta>0$, the following holds:
\begin{align*}
&|\mathcal{A}_{\zeta}^n(X)| \leq 2^{n(H(X)+\zeta')}\quad \mbox{and} \quad P_X^n\left[ \left( \mathcal{A}_{\zeta}^n(X) \right)^c \right]
\leq 2|\mathcal{X}| \exp \left( -2n \left(\frac{\zeta}{|\mathcal{X}|} \right)^2 \right),
\end{align*}
where $\zeta' \triangleq -\frac{\zeta}{|\mathcal{X}|}\sum_{a \in \mathcal{X}} \log P_X(a)$. 
\end{Lemma}
The proof follows from the application of the Chernoff inequality and the union bound, and is omitted for brevity.
\begin{Lemma}
\label{Lem:8}
Given a probability space $(\mathcal{X} \times \mathcal{Y},P_{X,Y})$,  parameters $ \zeta \in [0,|\mathcal{X}| P_{\max}]$ and $\delta>0$, blocklength $n \in \mathbb{N}$, and a vector $x^n \in \mathcal{A}^{n}_{\zeta}(X)$, we have
\begin{align}
    |\mathcal{A}_{\delta}^{n}(Y|x^n)| \geq 2^{n(H(Y|X)- \delta_1-\zeta_2)}
\alpha(\delta), \end{align}
where $P_{\max}=\max_{a \in \mathcal{X}} P_X(a)$,
$\delta_1 \triangleq -\frac{\delta}{|\mathcal{Y}|} \sum_{a \in \mathcal{X}} \sum_{b \in \mathcal{Y}}  \log P_{Y|X}(b|a)$, 
$\zeta_2 \triangleq \frac{\zeta}{|\mathcal{X}|} \sum_{a \in \mathcal{X}} H(Y|X=a)$, 
and $\alpha(\delta) \triangleq  \left[1-2|\mathcal{Y}| \exp \left\{ -\frac{n}{P_{\max} } \left( \frac{\delta}{|\mathcal{Y}|} 
\right)^2 \right\} \right]^{|\mathcal{X}|}$.
\end{Lemma}}
\begin{proof}
\textcolor{black}{We first provide a lower bound on the probability of the conditional typical set as follows:
\begin{align*}
P_{Y|X}^n (\mathcal{A}_{\delta}^n(Y|x^n)|x^n) &= P_{Y|X}^n \left[ \left.
\bigcap_{a \in \mathcal{X}} \bigcap_{b \in \mathcal{Y}} \left| \frac{1}{n} N(a,b|x^n,Y^n)-\frac{1}{n} N(a|x^n) P_{Y|X}(b|a) \right| \leq \frac{\delta}{|\mathcal{Y}|} 
\right| x^n \right] \\
&=\prod_{a \in \mathcal{X}} P_{Y|X}^n \left[ \left.
\bigcap_{b \in \mathcal{Y}} \left| \frac{1}{n} N(a,b|x^n,Y^n)-\frac{1}{n} N(a|x^n) P_{Y|X}(b|a) \right| \leq \frac{\delta}{|\mathcal{Y}|} 
\right| x^n \right] \\
&\stackrel{(a)}{\geq} \prod_{a \in \mathcal{X}} \left[ 1-2|\mathcal{Y}| \exp \left\{ 
-\frac{2}{N(a|x^n)} \left( \frac{\delta n}{|\mathcal{Y}|} \right)^2
\right\} \right], \\
&\stackrel{(b)}{\geq}\prod_{a \in \mathcal{X}} \left[ 1-2|\mathcal{Y}| \exp \left\{ 
-\frac{n}{P_{\max} } \left( \frac{\delta }{|\mathcal{Y}|} \right)^2
\right\} \right],
\end{align*}
where (a) follows from the union bound and the Chernoff inequality: For any i.i.d. sequence $Z_1,Z_2,\ldots,$ of Bernoulli($p$) random variables we have  
\[
P \left[ \left|\frac{1}{n} \sum_{i=1}^n Z_i -p \right| \geq \delta  \right] \leq 
2 \exp(-2n \delta^2 ),
\]
and (b) follows from the inequality: $N(a|x^n) \leq n \left( P_X(a)+\frac{\zeta}{|\mathcal{X}|} \right) \leq 2n P_{\max}$, which follows using the condition on $\zeta$, and the fact that $x^n \in \mathcal{A}_{\zeta}^n(P_X)$. 
Next we provide an upper bound on the probability of conditional typical set. 
\begin{align*}
   P_{Y|X}^n (\mathcal{A}_{\delta}^n(Y|x^n)|x^n) &= \sum_{y^n \in \mathcal{A}_{\delta}^n(Y|x^n) }
   P_{Y|X}^n(y^n|x^n) \\
   &= \sum_{y^n \in \mathcal{A}_{\delta}^n(Y|x^n) }
   \exp_2 \left[\sum_{a \in \mathcal{X}} \sum_{b \in \mathcal{Y}} N(a,b|x^n,y^n) \log P_{Y|X}(b|a)  \right] \\
   &\stackrel{(a)}{\leq} | \mathcal{A}_{\delta}^n(Y|x^n)| \textcolor{black}{2^{-n(\sum_{a\in \mathcal{X}} \frac{N(a|x^n)}{n}  H(Y|X=a)+\delta_1)}} \\
   &\stackrel{(b)}{\leq} | \mathcal{A}_{\delta}^n(Y|x^n)| 2^{-n(H(Y|X)-\zeta_2-\delta_1)},
\end{align*}
where (a) follows from the definition of a conditional typical sequence, and (b) follows from the fact that  $x^n \in \mathcal{A}_{\zeta}^n(X)$. The statement of the lemma follows by combining the lower and upper bounds. } 
\end{proof}

\textcolor{black}{We proceed to describe the coding scheme achieving the bounds described in the statement of the proposition. Consider a source $(\mathcal{S},\mathcal{W},P_{S_1}, d^n_{S_1,W})$, a blocklength $n\in \mathcal{B}(\epsilon)$, and a parameter
$\tau\in (\frac{1}{\sqrt{n}}\sigma,\sigma')$. We use the following source coding scheme:
\\\noindent \textbf{Codebook Generation:} Let 
$\Theta= \ceil{2^{n(I(W;S_1)+\theta_n(\tau)-\frac{1}{n})}}$ and $\hat{\tau}= \tau(|\mathcal{S}|+|\mathcal{W}|)$. Pick $\Theta$ codewords from the set $\mathcal{A}_{\hat{\tau}}^n(W)$ randomly, uniformly, and with replacement to construct the codebook. Let the $i$th codeword be denoted by $W^n(i), i\in [1,\Theta]$.
\\\noindent\textbf{Encoding:} Given a source sequence $S_1^n$, the encoder produces the index $e(S_1^n)\in [1,\Theta]$ such that:
\begin{align*}
    e(S_1^n)= \min\{i: (S_1^n,W^n(i))\in \mathcal{A}_{\hat{\tau}}^n(S_1,W)\}.
\end{align*}
If such an index does not exist, the encoder produces $e(S_1^n)=1$. 
\\\noindent\textbf{Decoding:} Given the index $i$, the decoder reproduces $W^n(i)$.}

\textcolor{black}{The resulting rate is $\frac{1}{n}\log{\Theta}\leq  I(W;S_1)+\theta_n(\tau)$. It remains to be shown that the condition $P(d^n_{S_1,W}(S_1^n,Q_n(S_1^n))>\phi)\leq \phi'$ in Definition \ref{def:Q} is satisfied. Note that if there exists an index $i\in [1,\Theta]$ such that $(S_1^n,W^n(i))\in \mathcal{A}_{\hat{\tau}}^n(S_1,W)$ then, by definition we have $d^n_{S_1,W}(S_1^n,Q_n(S_1^n))\leq \phi$. So,  $P(d^n_{S_1,W}(S_1^n,Q_n(S_1^n))>\phi)\leq P(\nexists i\in [1,\Theta]: (S_1^n,W^n(i))\in \mathcal{A}_{\hat{\tau}}^n(S_1,W))$. Consequently, it suffices to show that $P(\nexists i\in [1,\Theta]: (S_1^n,W^n(i))\in \mathcal{A}_{\hat{\tau}}^n(S_1,W))\leq \phi'$.
We have:}
\begin{align*}
   \textcolor{black}{P(\nexists i\in [1,\Theta]: (S_1^n,W^n(i))\in \mathcal{A}_{\hat{\tau}}^n(S_1,W))}&
  \stackrel{(a)}
  {\leq}
  \textcolor{black}{
    P(\nexists i\in [1,\Theta]: (S_1^n,W^n(i))\in \mathcal{A}_{\hat{\tau}}^n(S_1,W), S_1^n\in \mathcal{A}_{\tau}^n(S_1))
    + P(S_1^n\notin \mathcal{A}_{\tau}^n(S_1))}
    \\
    &\stackrel{(b)}{\leq}
    \textcolor{black}{\sum_{s^n\in \mathcal{A}_\tau^n (S_1)}
    P^n_{S_1}(s^n)\Big(1- \frac{|\{w^n:(s^n,w^n)\in \mathcal{A}_{\hat{\tau}}(S_1,W)\}|}{|\mathcal{A}_{\hat{\tau}}^n(W)|} \Big)^{\Theta}+P(S_1^n\notin \mathcal{A}_{\tau}^n(S_1))}
    \\&
    \stackrel{(c)}{\leq} 
    \textcolor{black}{
    \sum_{s^n\in \mathcal{A}_\tau^n (S_1)}
    P^n_{S_1}(s^n)\Big(1- \frac{2^{n(H(W|S_1)- \tau_1-\tau_2)} 
\alpha(\tau)}{|\mathcal{A}_{\hat{\tau}}^n(W)|} \Big)^{\Theta}+P(S_1^n\notin \mathcal{A}_{\tau}^n(S_1))}
\\&
\stackrel{(d)}{\leq} 
\textcolor{black}{
    \sum_{s^n\in \mathcal{A}_\tau^n (S_1)}
    P^n_{S_1}(s^n)\Big(1- \frac{2^{n(H(W|S_1)- \tau_1-\tau_2)} 
\alpha(\tau)}{2^{n(H(S_1)+\tau')}} \Big)^{\Theta}+P(S_1^n\notin \mathcal{A}_{\tau}^n(S_1))}
\\& \stackrel{(e)}{\leq}
\textcolor{black}{
\sum_{s^n\in \mathcal{A}_\tau^n (S_1)}
    P^n_{S_1}(s^n)\Big(1- \frac{2^{n(H(W|S_1)- \tau_1-\tau_2)} 
\alpha(\tau)}{2^{n(H(S_1)+\tau')}} \Big)^{\Theta}+
 2|\mathcal{S}| \exp \left( -2n \left(\frac{\tau}{|\mathcal{S}|} \right)^2 \right)}
 \\& \stackrel{(f)}{\leq}
 \textcolor{black}{
 \sum_{s^n\in \mathcal{A}_\tau^n (S_1)}
    P^n_{S_1}(s^n)\exp({-\Theta\alpha(\tau)2^{-n(I(S_1;W)+\tau_1+\tau_2+\tau')}})+
 2|\mathcal{S}| \exp \left( -2n \left(\frac{\tau}{|\mathcal{S}|} \right)^2 \right)}
 \\&{\leq}
 \textcolor{black}{
 \exp({-2^{n(I(W;S_1)+\theta_n(\tau)-\frac{1}{n})}\alpha(\tau)2^{-n(I(S_1;W)+\tau_1+\tau_2+\tau')}})+
 2|\mathcal{S}| \exp \left( -2n \left(\frac{\tau}{|\mathcal{S}|} \right)^2 \right),}
\end{align*}
\textcolor{black}{
where $\tau_1 \triangleq \frac{\tau}{|\mathcal{W}|}\sum_{a,b\in \mathcal{S}\times \mathcal{W}}\log{\frac{1}{P_{W|S_1}(b|a)}}$, $\tau_2\triangleq \frac{\tau}{|\mathcal{S}|}\sum_{a\in \mathcal{S}}H(W|S_1=a)$, and $\tau'= \frac{\hat{\tau}}{|\mathcal{S}|}\sum_{a\in \mathcal{S}}\log{\frac{1}{P_{S_1}(a)}}$, and (a) follows from the law of total probability, (b) follows from the fact that the codewords are chosen independently and with equal probability from $\mathcal{A}_{\hat{\tau}}^n(W)$, (c) follows from Lemma \ref{Lem:8}, (d) and (e) follow from Lemma \ref{Lem:6}, and (f) follows from $(1-y)^n\leq e^{-yn}$.} 
\textcolor{black}{Note that using the limits on $\tau  $ and $n$, it follows that $\alpha(\tau)\geq  \frac{1}{2}^{|\mathcal{S}|}$. So, 
\begin{align*}
    P(\nexists i\in [1,\Theta]: (S_1^n,W^n(i))\in \mathcal{A}_{\hat{\tau}}^n(S_1,W))
    \leq 
    \exp({-2^{n(I(W;S_1)+\theta_n(\tau)-\frac{1}{n}-\frac{1}{n}|\mathcal{S}|-I(S_1;W)-\tau_1-\tau_2-\tau')}})+
 2|\mathcal{S}| \exp \left( -2n \left(\frac{\tau}{|\mathcal{S}|} \right)^2 \right)
\end{align*}
Note that $\theta_n(\tau)$ defined in the proposition statement can be rewritten as follows:
\begin{align*}
\theta_n(\tau)=
    \frac{1}{n}\log\left(2n\frac{\tau^2}{|\mathcal{S}|^2}- \ln{2|\mathcal{S}|}\right)+\tau_1+\tau_2+\tau'+\frac{1}{n}|\mathcal{S}|+\frac{1}{n}.
\end{align*}
Hence,
\begin{align*}
    P(\nexists i\in [1,\Theta]: (S_1^n,W^n(i))\in \mathcal{A}_{\hat{\tau}}^n(S_1,W))
    &\leq 
    \exp({-2^{\log\left(2n\frac{\tau^2}{|\mathcal{S}|^2}- \ln{2|\mathcal{S}|}\right)}})+
 2|\mathcal{S}| \exp \left( -2n \left(\frac{\tau}{|\mathcal{S}|} \right)^2 \right)
 \\& 
 =
    \exp\left({-\left(2n\frac{\tau^2}{|\mathcal{S}|^2}- \ln{2|\mathcal{S}|}\right)}\right)+
 2|\mathcal{S}| \exp \left( -2n \left(\frac{\tau}{|\mathcal{S}|} \right)^2 \right)
 \\&=  4|\mathcal{S}| \exp \left( -2n \left(\frac{\tau}{|\mathcal{S}|} \right)^2 \right)=\phi'.
\end{align*}}

\section{Proof of Claim \ref{Claim:1}}
\label{App:Claim1}
\textcolor{black}{First, we argue that for any $i\in [1,m]$, the tuple $(\widetilde{X}_1(j,i),\widetilde{X}_2(j,i),\widetilde{W}'_1(j,i),\widetilde{W}_2(j,i))$ is a function of \\$\left(X_1([1,n],i), X_2([1,n],i), T([1,n],i)\right)$.  This is true since $\widetilde{W}'_1([1,n],i)=W'_1(\pi_i([1,n]),i)$ which from \eqref{eq:s20} is a function of $T([1,n],i)$ and $Q(X_1([1,n],i))$, similarly, $\widetilde{W}_2([1,n],i)=W_2(\pi_i([1,n]),i)$ is a function of $W_2([1,n],i)=Q(X_2([1,n],i))$. Define $\overline{A}_i=\left(X_1([1,n],i), X_2([1,n],i), T([1,n],i)\right)$. $A^{m}$ is a vector of independent random tuples since $(X_1^{n,m}, X_2^{n,m},T^{n,m})$ is a vector of i.i.d. triples. So,  $(\widetilde{X}_1(j,[1,m]),\widetilde{X}_2(j,[1,m]),\widetilde{W}'_1(j,[1,m]),\widetilde{W}_2(j,[1,m]))$ is a vector of independent random variables. Next, we show that these random variables are identically distributed. For arbitrary $i,l\in[1,m]$, $j \in [1,n]$, and $a_1,a_2,b_1,b_2\in \mathcal{X}_1\times\mathcal{X}_2\times\mathcal{W}\times\mathcal{W}$, we have:
 \begin{align*}
 P'&_{\widetilde{X}_1(j,i),\widetilde{X}_2(j,i),\widetilde{W}'_1(j,i),\widetilde{W}_2(j,i)}(a_1,a_2,b_1,b_2)
 =\sum_{k=1}^nP(\pi_{i}^{-1}(j)=k)P'_{X_1(k,i),X_2(k,i),{W}'_1(k,i),{W}_2(k,i)}(a_1,a_2,b_1,b_2)\\
 &=\frac{1}{n}\sum_{k=1}^n
 P'_{X_1(k,i),X_2(k,i),{W}'_1(k,i),{W}_2(k,i)}(a_1,a_2,b_1,b_2) \\
&=\frac{1}{n}\sum_{k=1}^n\sum_{x_1^{n},x_2^{n}:x_i(k)=a_i}P_{X_1([1,n],i),X_2([1,n],i)}(x_1^n,x^n_2)P'_{{W}'_1(k,i)|X_1([1,n],i)}(b_1|x_1^n)P'_{W_2(k,i)|X_2([1,n],i)}(b_2|x_2^n)\\
&=\frac{1}{n} \sum_{k=1}^n \sum_{x_1^{n},x_2^{n}:x_i(k)=a_i}P_{X_1^n,X_2^n}(x_1^n,x^n_2)P(W'_1(k,i)=b_1,W_2(k,i)=b_2|X_1([1,n],i)=x_1^n,X_2([1,n],i)=x_2^n)\\
&=\frac{1}{n} \sum_{k=1}^n \sum_{x_1^{n},x_2^{n}:x_i(k)=a_i}P'_{X_1^n,X_2^n}(x_1^n,x^n_2)P(W'_1(k,l)=b_1,W_2(k,l)=b_2|X_1([1,n],l)=x_1^n,X_2([1,n],l)=x_2^n)
\\
&=P'_{\widetilde{X}_1(j,l),\widetilde{X}_2(j,l),\widetilde{W}'_1(j,l),\widetilde{W}_2(j,l)}(a_1,a_2,b_1,b_2).
\end{align*}
Also, for any $j,l\in [1,n]$, $i\in [1,m]$, and $a_1,a_2,b_1,b_2\in \mathcal{X}_1\times\mathcal{X}_2\times\mathcal{W}\times\mathcal{W}$ we have:
\begin{align*}
  P'_{\widetilde{X}_1(j,i),\widetilde{X}_2(j,i),\widetilde{W}'_1(j,i),\widetilde{W}_2(j,i)}(a_1,a_2,b_1,b_2)&
 =\sum_{k=1}^nP(\pi_{j}^{-1}(i)=k)P'_{X_1(k,i),X_2(k,i),{W}'_1(k,i),{W}_2(k,i)}(a_1,a_2,b_1,b_2)\\
 &=
 \sum_{k=1}^nP(\pi_{l}^{-1}(i)=k)P'_{X_1(k,i),X_2(k,i),{W}'_1(k,i),{W}_2(k,i)}(a_1,a_2,b_1,b_2)
 \\
 &=  P'_{\widetilde{X}_1(l,i),\widetilde{X}_2(l,i),\widetilde{W}'_1(l,i),\widetilde{W}_2(l,i)}(a_1,a_2,b_1,b_2).
\end{align*}
Note that $P'_{\widetilde{X}_1(j,l),\widetilde{X}_2(j,l),\widetilde{W}'_1(j,l),\widetilde{W}_2(j,l)}$ is equal to  $P'_{{X}_1,{X}_2,{W}'_1,{W}_2}$ as defined in Equation \eqref{eq:9}. }

\section{Proof of Claim \ref{Claim:2}}
\label{App:Claim2}
\textcolor{black}{
1) Note that from \eqref{eq:distpp}, we have $P'_{X_1,X_2,W'_1,U_1}= P'_{X_1,X_2,W'_1}P_{U_1|X_1,W}$. So, it suffices to show that $P'_{X_1,X_2,W'_1}=P_{X_1,X_2}P_{W|S_1}$. This is true since
\begin{align*}
P'_{X_1,X_2,W'_1}&\stackrel{(a)}{=}\frac{1}{n}\sum_{i=1}^n P'_{X_1(i),X_2(i),W'_1(i)}=\frac{1}{n}\sum_{i=1}^n P'_{X_2(i)|X_{1}(i)}P'_{X_1(i)W'_1(i)}\\
&\stackrel{(b)}{=}P_{X_2|X_{1}}\frac{1}{n}\sum_{i=1}^n P'_{X_1(i)W'_1(i)}
\stackrel{(c)}{=} P_{X_2|X_1}\frac{1}{n}\sum_{i=1}^nP'_{X_1(i)}P'_{W'_1(i)|S_1(i)}
\stackrel{(d)}{=}P_{X_1,X_2}P_{W_1|S_1},
\end{align*}
where in (a) we have used the definition given in \eqref{eq:9}, in (b) and (c)  we have used the Markov chain $W'_1(i) - S_1(i) - X_1(i) - X_2(i)$, which holds because $P_{{W'_1}^n,S_1^n,X_1^n,X_2^n}= P_{{W'_1}^n|S_1^n}\prod_{i=1}^n P_{X_1(i),S_1(i),X_2(i)}$,
and (d) follows from our choice of $P'_{T|S_1}$ in Step 2, where the objective was to make $P_{S_1,W}= P'_{S_1(I),W'_1(I)}$. Condition 2) is true by construction of $P'$ as given in \eqref{eq:distpp}.  
\\3) We have:
\begin{align*}
 P'(W'_1=W_2,S_1=S_2)
 &\geq\frac{1}{n}\sum_{i=1}^n\sum_{s^n\in \mathcal{S}^n}P(S_1^n=S_2^n=s^n)P'(W'_1(i)=W_2(i)|S^n_1=S^n_2=s^n)\\\
 &\stackrel{(a)}{\geq} 
 \frac{1}{n}\sum_{i=1}^n\sum_{s^n\in \mathcal{S}^n} P(S^n_1=S^n_2=s^n)P'(T(i)=0|S_1(i)=s)P'(W_1(i)=W_2(i)|S_1^n=S_2^n=s^n)
 \\&= \sum_{s^n\in \mathcal{S}^n}P(S^n_1=S^n_2=s^n)P'(T(i)=0|S_1(i)=s_i)
  \\
 &\geq p(\tau) P(S_1^n=S_2^n)
= p(\tau) (1-\delta_n),
\end{align*}
where in (a) we have used the Markov chain $T(i) - S_1(i) -  X_1^n,X_2^n$.
Condition  4) is true by construction of $P'$ as given in \eqref{eq:distpp}.
\\5) Note that:
\begin{align*}
    &P'_{X_1,X_2,W'_1,W_2}(x_1,x_2,w_1,w_2)= \frac{1}{n} \sum_{i=1}^n P'_{X_1(i), X_2(i), W'_1(i), W_2(i)}(x_1,x_2,w_1,w_2)
    \\&= 
    \frac{1}{n} \sum_{i=1}^n\sum_{x_1(j), x_2(j), j\neq i} P'_{X_1^n, X_2^n, W'_1(i), W_2(i)}(x^n_1,x^n_2,w_1,w_2)
    \\&= 
    \frac{1}{n} \sum_{i=1}^n\sum_{x_1(j), x_2(j), j\neq i}
    P_{S_1^n,S_2^n}(s_1^n,s_2^n)
    P_{X^n_1,X_2^n|S_1^n,S_2^n}(x_1^n,x_2^n|s_1^n,s_2^n)
    P'_{ W'_1(i), W_2(i)|S_1^n,S_2^n}(w_1,w_2|s_1^n,s_2^n)
    \\&  = 
    \frac{1}{n} \sum_{i=1}^n\sum_{x_1(j), x_2(j), j\neq i}
\left[    \prod_{k=1}^n
    P_{S_1(k),S_2(k)}(s_1(k),s_2(k))P_{X_1(k),X_2(k)|S_1(k),S_2(k)}(x_1(k),x_2(k)|s_1(k),s_2(k)) \right] \\
    & \hspace{1in} 
    \times P'_{ W'_1(i), W_2(i)|S_1^n,S_2^n}(w_1,w_2|s_1^n,s_2^n)
    \\&= 
    \frac{1}{n}
    \sum_{i=1}^n
    P_{X_1(i),X_2(i)|S_1(i),S_2(i)}(x_1,x_2|s_1,s_2)
    \sum_{s_1(j), s_2(j), j\neq i}
    \left[ \prod_{k=1}^n
    P_{S_1(k),S_2(k)}(s_1(k),s_2(k)) \right]
    P'_{ W'_1(i), W_2(i)|S_1^n,S_2^n}(w_1,w_2|s_1^n,s_2^n)
    \\&=\frac{1}{n}
    \sum_{i=1}^n
    P_{X_1(i),X_2(i)|S_1(i),S_2(i)}(x_1,x_2|s_1,s_2)
    P'_{ W'_1(i), W_2(i),S_1(i),S_2(i)}(w_1,w_2,s_1,s_2),
\\&=
    P_{X_1,X_2|S_1,S_2}(x_1,x_2|s_1,s_2)
    \frac{1}{n} \sum_{i=1}^n
    P'_{ W'_1(i), W_2(i),S_1(i),S_2(i)}(w_1,w_2,s_1,s_2),
\end{align*}
where we have defined $x_1(i)=x_1, x_2(i)=x_2, s_1(i)=s_1$, $s_2(i)=s_2$, and $s_j(i)=f_j(x_j(i)), i\in [n], j\in \{1,2\}$.}



\section{Proof of Lemma \ref{prop:mutI}}
\label{App:prop:mutI}

\textcolor{black}{\\1) 
\begin{align*}
 &v\geq V(P_{X,Y,Z},Q_{X,Y,Z})=\frac{1}{2}\sum_{x,y,z}\big|P_{X,Y,Z}(x,y,z)-Q_{X,Y,Z}(x,y,z)\big|
 \\&\geq
 \frac{1}{2}\sum_{x,y}\big|\sum_{z}(P_{X,Y,Z}(x,y,z)-Q_{X,Y,Z}(x,y,z))\big|=V(P_{X,Y},Q_{X,Y}).
\end{align*}
2) Note that $I_P(X;Y)= H_P(X)+H_P(Y)-H_P(X,Y)$. So:
\begin{align*}
 &|I_P(X;Y)-I_Q(X;Y)\big|= \big|H_P(X)+H_P(Y)-H_P(X,Y)-H_Q(X)-H_Q(Y)+H_P(X,Y)\big|\\
 &\leq \big|H_P(X)-H_Q(X)\big|+ \big|H_P(Y)-H_Q(Y)\big|+\big|H_P(X,Y)-H_Q(X,Y)\big|\\
 &\leq 4\left(h_b(v)+v\log {|\mathcal{A}|}\right),
\end{align*}
where the last step uses Lemma \ref{prop:RY} and the fact that $v\geq V(P_X,Q_X)$ and $v\geq V(P_Y,Q_Y)$ from 1).
\\3) Follows from 1) and 2) and the fact that $I(X;Y|Z)\leq I(X,Z;Y)-I(Z;Y)$.
}

\section{Proof of Claim \ref{claim:RY}}
\label{App:Claim3}
\textcolor{black}{
 We have:
 \begin{align*}
 2V(P',Q)&=\sum_{w_1,w_2,x_1,x_2,u_1,u_2}\big|P'_{W'_1,W_2,X_1,X_2,U_1,U_2}(w_1,w_2,x_1,x_2,u_1,u_2)-Q_{W,W,X_1,X_2,U_1,U_2}(w_1,w_2,x_1,x_2,u_1,u_2)\big|\\
 &= \sum_{w_1,w_2,x_1,x_2,u_1,u_2}Q_{U_1|W,X_1}(u_1|w_1,x_1)Q_{U_2|W,X_2}(u_2|w_2,x_2)\times
 \\&
 \big|P'_{W'_1,W_2,X_1,X_2}(w_1,w_2,x_1,x_2)-Q_{W,W,X_1,X_2}(w_1,w_2,x_1,x_2)\big|
\end{align*}
Also, we have:
\begin{align*}
&\big|P'_{W'_1,W_2,X_1,X_2}(w_1,w_2,x_1,x_2)-Q_{W,W,X_1,X_2}(w_1,w_2,x_1,x_2)\big|
\\&=
 P_{X_1,X_2|S_1,S_2}(x_1,x_2|s_1,s_2)\big|P'_{W'_1,W_2,S_1,S_2}(w_1,w_2,s_1,s_2)-Q_{W,W,S_1,S_2}(w_1,w_2,s_1,s_2)\big|,
\end{align*}
for all $x_1,x_2,w_1,w_2\in \mathcal{X}_1\times \mathcal{X}_2\times \mathcal{W}\times \mathcal{W}$, where we have used the Markov chain
$(W'_1, W_2) - (S_1,S_2) - (X_1,X_2)$ for $P'$ from Theorem \ref{thm: MCDSC22}, and the Markov chain $W- S_1 - (X_1,X_2)$ for $Q$ from Theorem \ref{thm:MC}. Furthermore,
\begin{align*}
 &\big|P'_{W'_1,W_2,S_1,S_2}(w_1,w_2,s_1,s_2)-Q_{W,W,S_1,S_2}(w_1,w_2,s_1,s_2)\big|
 \\&=\big|P'_{W'_1,S_1}(w_1,s_1)P'_{W_2,S_2|W'_1,S_1}(w_2,s_2|w_1,s_1)-\mathbbm{1}(w_1=w_2)Q_{W,S_1}(w_1,s_1)Q_{S_2|W,S_1}(s_2|w_1,s_1)\big|
\\&=Q_{W,S_1}(w_1,s_1)\big|P'_{W_2,S_2|W'_1,S_1}(w_2,s_2|w_1,s_1)-\mathbbm{1}(w_1=w_2)Q_{S_2|W,S_1}(s_2|w_1,s_1)\big|,
\end{align*}
for all $s_1,s_2,w_1,w_2\in \mathcal{S}_1\times \mathcal{S}_2\times \mathcal{W}\times \mathcal{W}$, and,
\begin{align*}
&\big|P'_{W_2,S_2|W'_1,S_1}(w_2,s_2|w_1,s_1)-\mathbbm{1}(w_1=w_2)Q_{S_2|W,S_1}(s_2|w_1,s_1)\big|
\\&\leq \big|P'_{W_2,S_2|W'_1,S_1}(w_2,s_2|w_1,s_1)-\mathbbm{1}(w_1=w_2, s_1=s_2)\big|
+\big|\mathbbm{1}(w_1=w_2, s_1=s_2)-\mathbbm{1}(w_1=w_2)Q_{S_2|W,S_1}(s_2|w_1,s_1)\big|.
\end{align*}
So:
\begin{align*}
 &2V(P'_{W'_1,W_2,X_1,X_2,U_1,U_2},Q_{W,W,X_1,X_2,U_1,U_2})\leq
 \\& \sum_{w_1,w_2,x_1,x_2,u_1,u_2}Q_{U_1|W,X_1}(u_1|w_1,x_1)Q_{U_2|W,X_2}(u_2|w_2,x_2)P_{X_1,X_2|S_1,S_2}(x_1,x_2|s_1,s_2)Q_{W,S_1}(w_1,s_1) \times
 \\&\left( \big|P'_{W_2,S_2|W'_1,S_1}(w_2,s_2|w_1,s_1)-\mathbbm{1}(w_1=w_2, s_1=s_2)\big|
+\big|\mathbbm{1}(w_1=w_2, s_1=s_2)-\mathbbm{1}(w_1=w_2)Q_{S_2|W,S_1}(s_2|w_1,s_1)\big|\right)
\\&= \sum_{w_1,w_2,x_1,x_2}P_{X_1,X_2|S_1,S_2}(x_1,x_2|s_1,s_2)Q_{W,S_1}(w_1,s_1)\cdot \big|P'_{W_2,S_2|W'_1,S_1}(w_2,s_2|w_1,s_1)-\mathbbm{1}(w_1=w_2, s_1=s_2)\big|
 \\&
+ \sum_{w,x_1,x_2}P_{X_1,X_2|S_1,S_2}(x_1,x_2|s_1,s_2)Q_{W,S_1}(w,s_1)\cdot \big|\mathbbm{1}( s_1=s_2)-Q_{S_2|W,S_1}(s_2|w,s_1)\big|
\\&=\sum_{w_1,w_2,s_1,s_2} Q_{W,S_1}(w_1,s_1)\big|P'_{W_2,S_2|W'_1,S_1}(w_2,s_2|w_1,s_1)-\mathbbm{1}(w_1=w_2, s_1=s_2)\big|
\\&+ \sum_{w,s_1,s_2}Q_{W,S_1}(w,s_1)\big|\mathbbm{1}( s_1=s_2)-Q_{S_2|W,S_1}(s_2|w,s_1)\big|.
\end{align*}
Consider the first term:
\begin{align}
&\nonumber \sum_{w_1,w_2,s_1,s_2} Q_{W,S_1}(w_1,s_1)\big|P'_{W_2,S_2|W'_1,S_1}(w_2,s_2|w_1,s_1)-\mathbbm{1}(w_1=w_2, s_1=s_2)\big|\\
&\nonumber=\sum_{\substack{w_1,w_2,s_1,s_2\\ (w_1,s_1)=(w_2,s_2)}} Q_{W,S_1}(w_1,s_1)\big|P'_{W_2,S_2|W'_1,S_1}(w_2,s_2|w_1,s_1)-1\big|
+\sum_{\substack{w_1,w_2,s_1,s_2\\ (w_1,s_1)\neq(w_2,s_2)}} Q_{W,S_1}(w_1,s_1)\big|P'_{W_2,S_2|W'_1,S_1}(w_2,s_2|w_1,s_1)\big|\\
&\nonumber=\sum_{w_1,s_1} Q_{W,S_1}(w_1,s_1)(1-P'_{W_2,S_2|W'_1,S_1}(w_1,s_1|w_1,s_1))
+\sum_{\substack{w_1,w_2,s_1,s_2\\ (w_1,s_1)\neq(w_2,s_2)}} P'_{W'_1,W_2,S_1,S_2}(w_1,w_2,s_1,s_2)\\
&\nonumber=1-P'(W'_1=W_2, S_1=S_2)+P'( (W'_1,S_1)\neq (W_2,S_2))\\
&= 2(1-P'(W'_1=W_2, S_1=S_2))\leq  2\left(1-p(\tau)(1-\delta_n)\right). \label{eq:41}
\end{align}
The second term can be bounded from above  similarly:
\begin{align}
 \sum_{w,s_1,s_2}Q_{W,S_1}(w,s_1)\big|\mathbbm{1}( s_1=s_2)-Q_{S_2|W,S_1}(s_2|w,s_1)\big|= 2(1-P(S_1=S_2))
 \leq 2\epsilon.
 \label{eq:42}
\end{align}
Combining \eqref{eq:41} and \eqref{eq:42}, we get $V(P,Q)\leq 1-p(\tau)+p(\tau)\delta_n+\epsilon$.}

\section{Proof of Lemma \ref{Lem: Fin1}}
\label{App:Lem: Fin1}
\textcolor{black}{
 First, note that $\widetilde{S}([1,n],j)$ is a function of $(X([1,n],j),Z([1,n],j))$ for any $j\in [1,m]$. Since the sources are discrete memoryless, $\widetilde{S}([1,n],j)$ and $\widetilde{S}([1,n],j')$ are independent of each other for any $j,j'\in [1,m]$, where $j\neq j'$. 
It remains to show that $\widetilde{S}(i,j), (i,j)\in [1,n]\times [1,m]$ are identically distributed. We have:
\begin{align*}
P(\widetilde{S}(i,j)=1)&=P(X(\pi_j(i),j)+\widehat{V}(\pi_j(i),j))+Z(\pi_j(i),j)=1)\\
&\stackrel{(a)}{=}p\ast P(X(\pi_j(i),j)+\widehat{V}(\pi_j(i),j)=1)\\
&\stackrel{(b)}{=}p\ast \frac{1}{n}\sum_{i\rq{}=1}^n \mathbb{E}(w_H(X(i',j)+\widehat{V}(i',j)))\\
&\stackrel{(c)}{=}p\ast \delta'
\end{align*}
(a) holds since $Z$ is independent of $X$, and $X(\pi_j(i),j)+\widehat{V}(\pi_j(i),j)$ is a function of $X([1,n],[1,m])$, (b) holds since the choice of $\pi_i$ is independent of the source sequences, and (c) holds since  $C_f^{(n)}$ corresponds to the quantizer $Q_{n}^*$ with distortion $\delta'$. 
}
\end{appendices}


%
 \end{document}